%% file: main.tex
\documentclass[11pt]{article}

\usepackage[
	persons={Callum, Kevin, June, Nima},
	bibliosources={refs.bib},
	ifclass={{acmart,sigconf,sigconf}}
]{pomegranate}

\DeclareDelimiter{\Otilde}[\mathnormal{\widetilde{O}}]{\lparen}{\rparen}
\DeclareDelimiter{\Omegatilde}[\mathnormal{\widetilde{\Omega}}]{\lparen}{\rparen}
\DeclareDelimiter{\DKL}[{\mathcal{D}_{\operatorname{KL}}}]{\lparen}{\rparen}
\DeclareDelimiter{\dTV}[{\mathnormal{d}_{\operatorname{TV}}}]{\lparen}{\rparen}

\DeclareDocumentMathCommand{\tOmega}{}{\widetilde{\Omega}}
\DeclareDocumentMathCommand{\lam}{}{\lambda}
\DeclareDocumentMathCommand{\eps}{}{\epsilon}
\DeclareDocumentMathCommand{\veps}{}{\varepsilon}
\DeclareDocumentMathCommand{\hmu}{}{\hat{\mu}}
\DeclareDocumentMathCommand{\muiso}{}{\mu^{\operatorname{iso}}}
\DeclareDocumentMathCommand{\xset}{}{\mathcal{X}}
\DeclareDocumentMathCommand{\map}{}{\mathcal{M}}
\DeclareDocumentMathCommand{\cert}{}{\mathcal{C}}

\DeclareOperator{\diag}
\DeclareOperator{\tr}

\title{Quadratic Speedups in Parallel Sampling from Determinantal Distributions}
\date{}

\Tag{
	\author[1]{Nima Anari}
	\author[1]{Callum Burgess}
	\author[2]{Kevin Tian}
	\author[1]{Thuy-Duong Vuong}
	\affil[1]{Stanford University, \url{{anari,callumb,tdvuong}@stanford.edu}}
	\affil[2]{Microsoft Research, \url{tiankevin@microsoft.com}}
}

\Tag<sigconf>{
	\author{Nima Anari}
	\email{anari@cs.stanford.edu}
	\orcid{0000-0002-4394-3530}
	\affiliation{
	  \institution{Stanford University}
	  \city{Stanford}
	  \state{CA}
	  \country{USA}
	}
	
	\author{Callum Burgess}
	\email{callumb@stanford.edu}
	\orcid{0009-0008-3669-037X}
	\affiliation{
	  \institution{Stanford University}
	  \city{Stanford}
	  \state{CA}
	  \country{USA}
	}
	
	\author{Kevin Tian}
	\email{tiankevin@microsoft.com}
	\orcid{0000-0003-0109-9201}
	\affiliation{
	  \institution{Microsoft Research}
	  \city{Redmond}
	  \state{WA}
	  \country{USA}
	}
	
	\author{Thuy-Duong Vuong}
	\email{tdvuong@stanford.edu}
	\orcid{0000-0003-0271-9687}
	\affiliation{
	  \institution{Stanford University}
	  \city{Stanford}
	  \state{CA}
	  \country{USA}
	}
}

\Tag<sigconf>{
	\input{concepts}

	\copyrightyear{2023} 
	\acmYear{2023} 
	\setcopyright{acmlicensed}\acmConference[SPAA '23]{Proceedings of the 35th ACM Symposium on Parallelism in Algorithms and Architectures}{June 17--19, 2023}{Orlando, FL, USA}
	\acmBooktitle{Proceedings of the 35th ACM Symposium on Parallelism in Algorithms and Architectures (SPAA '23), June 17--19, 2023, Orlando, FL, USA}
	\acmPrice{15.00}
	\acmDOI{10.1145/3558481.3591104}
	\acmISBN{978-1-4503-9545-8/23/06}
	
	\settopmatter{printacmref=true}
}
\begin{document}
	\Tag<sigconf>{
		\begin{abstract}	\input{abstract}\end{abstract}
		\keywords{sampling, counting, rejection sampling, parallel algorithms, determinantal point process, planar perfect matchings}
		\maketitle
	}
	\Tag{
		\maketitle
		\begin{abstract}\input{abstract.tex}\end{abstract}
		\clearpage
	}
	
	\input{intro}
	\input{overview}
    \input{prelim}
    \input{symmetric}
    \input{ei}
    \input{planar}
	
	\Tag{
	    \input{hard}
	    \input{unsized-dpp}
	    \input{proofs}
    }
	
	\Tag<sigconf>{\input{acmacks}}
	
	\Tag<sigconf>{\balance}
	\PrintBibliography
\end{document}

%% file: concepts.tex
\ccsdesc[500]{Theory of computation~Generating random combinatorial structures}
\ccsdesc[500]{Theory of computation~Parallel algorithms}

%% file: abstract.tex
We study the problem of parallelizing sampling from distributions related to determinants:  symmetric, nonsymmetric, and partition-constrained determinantal point processes, as well as planar perfect matchings. For these distributions, the partition function, a.k.a.\ the count, can be obtained via matrix determinants, a highly parallelizable computation; Csanky proved it is in NC. However, parallel counting does not automatically translate to parallel sampling, as classic reductions between the two are inherently sequential. We show that a nearly quadratic parallel speedup over sequential sampling can be achieved for all the aforementioned distributions. If the distribution is supported on subsets of size $k$ of a ground set, we show how to approximately produce a sample in $\widetilde{O}(k^{\frac{1}{2} + c})$ time with polynomially many processors for any constant $c>0$. In the two special cases of symmetric determinantal point processes and planar perfect matchings, our bound improves to $\widetilde{O}(\sqrt k)$ and we show how to sample exactly in these cases.
		
As our main technical contribution, we fully characterize the limits of batching for the steps of sampling-to-counting reductions. We observe that only $O(1)$ steps can be batched together if we strive for exact sampling, even in the case of nonsymmetric determinantal point processes. However, we show that for approximate sampling, $\widetilde{\Omega}(k^{\frac{1}{2}-c})$ steps can be batched together, for any entropically independent distribution, which includes all mentioned classes of determinantal point processes. Entropic independence and related notions have been the source of breakthroughs in Markov chain analysis in recent years, so we expect our framework to prove useful for distributions beyond those studied in this work.

%% file: intro.tex
\section{Introduction}

Sampling and counting are intimately connected problems. For many classes of measures $\mu$ defined on a, most often exponentially large, space $X$, approximately sampling $x\in X$ with $\P{x}\propto \mu(x)$ and approximately computing the partition function $\sum_{x\in X} \mu(x)$ are polynomial-time reducible to each other \cite{jerrum1986random}. However, this equivalence breaks down for complexity classes below \Class{P}. For example, there is no known polylogarithmic-time parallel reduction between approximate counting and approximate sampling.

Motivated by the underexplored relationship between sampling and counting in the world of parallel algorithms, \textcite{anari2020sampling}, based on the earlier work of \textcite{teng1995independent}, raised the question of designing fast parallel \emph{samplers} for several distributions where counting, even exactly, was possible in polylogarithmic time and polynomial work, i.e., via \Class{NC} algorithms. The distributions in this challenge set all enjoy fast parallel counting algorithms because their partition functions, $\sum_{x}\mu(x)$, can be written as determinants, and determinants are computable in \Class{NC} \cite{csanky1975fast}. \textcite{anari2020sampling} solved one of these challenges and showed how to sample random arborescences in \Class{RNC}, completing the earlier work of \textcite{teng1995independent} on random spanning trees. However, these works are tailored to the random spanning tree and arborescence distributions, as they parallelize the sequential sampling algorithm of \textcite{Aldous90,Broder89}, and there is no known generalization of this algorithm to other distributions.

In this work, we study a general framework to improve the parallel efficiency of sampling-to-counting reductions. We build on the success of recently introduced techniques in the analysis of random walks and sampling algorithms, where combinatorial distributions are analyzed through the lens of high-dimensional expanders \cite{ALOV19,AL20,ALO20}. We show that under one notion of high-dimensional expansion named entropic independence \cite{AJKPV21}, see \cref{def:entropic-independence} for definition, sampling-to-counting reductions can be made sped up nearly quadratically on a \Class{PRAM}.

To set up notation, we consider combinatorial distributions defined on size $k$ subsets of a ground set of elements $\set{1,\dots,n}$, which we denote by an (unnormalized) measure $\mu:\binom{[n]}{k}\to\R_{\geq 0}$. We note that the choice of $\binom{[n]}{k}$ is a standard canonical form for the domain, and many other domains such as product spaces, can be transformed into $\binom{[n]}{k}$ \cite{ALO20}. We access the measure $\mu$ through a counting oracle. Given \emph{any}\footnote{By querying a $T$ of size exactly $k$, a counting query can also return $\mu(T)$.} set $T\subseteq [n]$, the oracle returns
\[ \sum\set*{\mu(S)\given S\in \binom{[n]}{k}, T\subseteq S}. \]
Our goal is to use the oracle and output a random set $S$ that approximately follows $\P{S}\propto \mu(S)$.

The classic reduction from sampling to counting \cite{jerrum1986random} proceeds by picking the $k$ elements of $S$, one at a time. In each step, conditioned on all previously chosen elements, marginals $\P_{S\sim \mu}{i\in S\given \text{previous choices}}$ of all remaining elements $i$ are computed and a new element is picked randomly with probability proportional to the conditional marginals. In each step, marginals can be computed via parallel calls to the counting oracle. However, this procedure is inherently sequential as the choice of each element affects the conditional marginals in future iterations. A parallel implementation of this reduction takes time $\Omega(k)$. The main question we study is:
\begin{quote}
    \textit{For which $\mu$ is there a faster parallel reduction from sampling to counting?}
\end{quote}

Our main result shows that for distributions $\mu$ which are entropically independent \cite{AJKPV21}, sampling-to-counting reductions can be sped up roughly quadratically. Throughout, we use $\Otilde{\cdot}$ to hide logarithmic factors in $n$ and failure probabilities; these factors primarily come from the parallel complexity of linear algebra (e.g., evaluating determinants and partition functions).

\begin{theorem}[Main, see \cref{thm:mainei} for formal version]\label{thm:main}
Let $\mu: \binom{[n]}{k} \to \R_{\geq 0}$ be entropically independent and assume that we have access to a counting oracle for $\mu$. For any constant $c > 0$ and any $\eps>0$, there exists an algorithm that can sample from a distribution within total variation distance $\eps$ of $\mu$ in $\Otilde*{\sqrt{k} \cdot \parens*{\frac k \eps}^{c}}$
parallel time using $(n/\eps)^{O(1/c)}$ machines in the \Class{PRAM} model.
\end{theorem}

As a corollary, we get improved parallel sampling for various classes of determinantal point processes: symmetric, non-symmetric, and partition-constrained. All of these distributions are known to be entropically independent \cite{AASV21,AJKPV21}. As an additional result, we show that by using different graph-separator-based techniques, a similar quadratic parallel speedup can be obtained for sampling planar perfect matchings, the only distribution in the challenge set of \cite{anari2020sampling} not covered by \cref{thm:main}; see \cref{sec:planar} for details.

\begin{remark}[Beyond determinantal distributions]
Various notions of high-dimensional expansion, and in particular, entropic independence, have recently resulted in breakthroughs in the analysis of Markov chains and sequential sampling algorithms \cite{AJKPV21,AJKPV21b}, but as far as we know, this is the first work to relate these notions to parallel sampling. Here we explore the applications of \cref{thm:main} to distributions related to determinants because the counting oracle for them can be implemented in \Class{RNC}. However, we suspect \cref{thm:main} can have applications beyond determinantal distributions in the future. As an example, for distributions whose partition functions do not have roots in certain regions of the complex plane, \textcite{barvinok2018approximating} devised efficient deterministic approximate counting algorithms, which have been refined by subsequent works \cite{patel2017deterministic}. These counting algorithms can be parallelized in many settings, as they involve enumerating a polynomial/quasi-polynomial number of small, logarithmic-sized, combinatorial substructures. Recent works \cite{AASV21,chen2021spectral} have shown that the absence of roots in the complex plane implies certain forms of high-dimensional expansion, including entropic independence \cite{AJKPV21}.
\end{remark}

\subsection{Determinantal distributions}

This work considers applications of \cref{thm:main} to various distributions $\mu$ defined based on determinants. Prior progress on designing parallel samplers for problems that enjoy determinant-based counting has been limited. \Textcite{teng1995independent} showed how to simulate random walks on a graph in parallel, which combined with the classic algorithm of \textcite{Aldous90,Broder89} yielded \Class{RNC} samplers for uniformly random spanning trees in a graph. \Textcite{anari2020sampling} extended this to sampling arborescences (directed spanning trees) in directed graphs. In this work, we tackle a much larger class of problems that enjoy determinant-based counting, namely variants of determinantal point processes.

Determinantal point processes (DPPs) have found many applications, such as data summarization \cite{gong2014largemargin,LB12}, recommender systems \cite{GartPK16,Wilhelm18}, neural network compression \cite{MS15}, kernel approximation \cite{LiJS16}, multi-modal output generation \cite{Elfeki19}, and randomized numerical linear algebra \cite{DM21}. Formally, a DPP on a set of items $[n]=\set{1,\dots,n}$ is a probability distribution over subsets $Y\subseteq [n]$ defined via an $n\times n$ matrix $L$ where probabilities are given (proportionally) by principal minors:
$\P{Y}\propto \det(L_{Y, Y})$. The partition function of such a distribution is simply $\det(L+I)$, which can be efficiently computed in parallel.

Note that for the distribution to be well-defined, all principal minors of $L$ have to be $\geq 0$. For symmetric $L$, that is $L = L^\intercal$, this is equivalent to $L$ being positive semi-definite (PSD). Symmetric DPPs, where $L = L^\intercal\succeq 0$ have received the most attention in the literature.
\begin{definition}[Symmetric DPP]
    Given a symmetric $n\times n$ matrix $L\succeq 0$, the symmetric DPP defined by $L$ is the probability distribution over subsets $Y\subseteq [n]$, where
    $ \P{Y}\propto \det(L_{Y, Y})$.
\end{definition}

Beyond (symmetric) determinantal point processes, our work provides sampling algorithms for a variety of other discrete distributions related to determinants: non-symmetric and partition-constrained DPPs, as well as planar perfect matchings. Next, we will outline each family of such distributions.

Recently, \cite{BWLGCG18,Gartrell2019LearningND,gartrell2020scalable} initiated the study of non-symmetric DPPs in applications and argued for their use because of increased modeling power. Symmetric DPPs necessarily exhibit strong forms of negative dependence \cite{BBL09}, which are unrealistic in some applications; on the other hand, non-symmetric DPPs can have positive correlations. Non-symmetric DPPs are characterized by a non-symmetric positive-definite matrix $L$, i.e., a matrix $L$ where $L+ L^\intercal \succeq 0$.

\begin{definition}\label{def:nPSD}
	A matrix $L\in \R^{n\times n}$ is \textit{non-symmetric positive semidefinite} (nPSD) if $L + L^\intercal \succeq 0$.  
\end{definition}

\begin{definition}[Non-symmetric DPP]
    Given an nPSD $n\times n$ matrix $L$, the non-symmetric DPP defined by it is the probability distribution over subsets $Y\subseteq [n]$ given by
    $\P{Y}\propto \det(L_{Y, Y})$.
\end{definition}

A related and more commonly used model related to DPPs is a $k$-DPP, where we constrain the cardinality of the sampled set $Y$ to be exactly $k$. In many applications, restricting to sets of a predetermined size is more desirable \cite{KT12}.

\begin{definition}[$k$-DPP]
    Given a PSD or nPSD matrix $L$, the $k$-DPP defined by it is the distribution of the corresponding determinantal point process restricted to only $k$-sized sets.
\end{definition}

A natural generalization of simple cardinality constraints is partition constraints \cite{KD16}. Partition constraints arise naturally when there is an inherent labeling or grouping of the ground set items that is not captured by the DPP kernel itself. Concretely, suppose that the ground set $[n]$ is partitioned into disjoint sets $V_1 \cup V_2 \cup \dots \cup V_r$, and we want to produce a subset $S$ with $c_1$ items from $V_1$, $c_2$ items from $V_2$ and so on. We define Partition-DPP as the corresponding conditioning of the DPP under these constraints on $S$. \Textcite{KD16} established that efficiently sampling and counting from Partition-DPPs is possible when the number of constraints $r$ is $O(1)$ and that counting is \Class{\#P}-hard without such restrictions on $r$ -- it includes, as a special case, the problem of computing mixed discriminants. Here, we will only study Partition-DPPs where the ensemble matrix $L$ is symmetric PSD and the number of constraints is $O(1)$. \Textcite{AASV21} showed that these distributions are entropically independent and that local Markov chains can be used to sample (sequentially) from these Partition-DPPs.
 \begin{definition}[Partition-DPP]
    Given a symmetric $n\times n$ matrix $L\succeq 0$ and a partitioning of $[n]=V_1\cup V_2\cup\cdots\cup V_r$ into $r=O(1)$ partitions together with $c_1,\dots,c_r\in \Z_{\geq 0}$, the Partition-DPP is the distribution of the DPP defined by $L$ restricted to sets $S$ that have $\card{S\cap V_i}=c_i$ for all $i$.
 \end{definition}

In this work, we establish as corollaries of \cref{thm:main}, a roughly quadratic parallel speedup in sampling from all aforementioned distributions. A crucial part of our algorithm relies on the existence of highly parallel counting oracles for these models. For example, for unconstrained DPPs, the partition function can be written as $\sum_{S} \det(L_{S, S})=\det(L+I)$, and this can be computed in \Class{NC} \cite{csanky1975fast}. For $k$-DPPs and Partition-DPPs, the partition function can be computed via polynomial interpolation \cite{KD16}, which is again highly parallelizable (by, e.g., solving linear systems of equations involving Vandermonde matrices). The entropic independence of all the aforementioned determinantal distributions discussed was established by prior works \cite{AASV21,AJKPV21}.

\begin{theorem}[Sampling from non-symmetric DPPs]\label{thm:nsdpp}
Let $L$ be a $n \times n$ non-symmetric PSD matrix, $\eps>0$, and $k \in [n]$.
\begin{enumerate}
    \item Let $\mu_k: \binom{[n]}{k} \to \R_{\geq 0}$ be the $k$-DPP defined by $L.$ For any constant $c > 0$, there exists an algorithm to approximately sample from within $\eps$ total variation distance of $\mu_k$ in $\Otilde*{\sqrt k(\frac k \eps)^c}$ parallel time using $(n/\eps)^{O(1/c)}$ machines.
    \item Let  $\mu: 2^{[n]} \to \R_{\geq 0}$ be the DPP defined by $L.$ For any constant $c > 0$, there exists an algorithm to approximately sample from within $\eps$ total variation distance of $\mu$ in $\Otilde*{\sqrt n(\frac n \eps)^c}$ parallel time using $(n/\eps)^{O(1/c)}$ machines.
\end{enumerate}
\end{theorem}

\begin{theorem}[Sampling from Partition-DPPs]\label{thm:pcdpp}
Let $L$ be a $n \times n$ symmetric PSD matrix. Let $r = O(1)$, and let $V_1\cup \dots \cup V_r = [n]$ be a partition of $[n]$ together with integers $t_1,\dots,t_r$. Let $k = \sum_{i \in [r]} t_i$. Let $\mu_{L; V,t}: 2^{[n]} \to \R_{\geq 0}$ be the DPP with partition constraints defined by
\[\mu_{L; V,t}(S) \propto \det(L_{S,S})\cdot \prod_{i=1}^r \1[\card{S \cap V_i} = t_i].\]
For any constant $c > 0$, there exists an algorithm to approximately sample from within $\eps$ total variation distance of $\mu_{L; V, t}$ in $\Otilde*{\sqrt k(\frac k \eps)^c}$ parallel time using $(n/\eps)^{O(1/c)}$ machines.
\end{theorem}

In the case of symmetric DPPs and symmetric $k$-DPPs, we improve \cref{thm:main} to obtain a parallel runtime of $\Otilde{\sqrt{k}}$. Our algorithms have a small chance $\delta$ of failure but, conditioned on success, they sample exactly from the desired distribution; this is desirable, as we can repeat the algorithm in the case of failure, to eventually obtain an exact sample from the distribution. \Cref{thm:sym DPP} is proven in \cref{sec:symmetric-dpp}.

\begin{theorem}[Sampling from symmetric DPPs] \label{thm:sym DPP}
 Let $L$ be a $n \times n$ symmetric PSD matrix, $k \in [n]$, and $\delta \in (0, 1)$.
 
\begin{enumerate} 
    \item Let $\mu_k: \binom{[n]}{k} \to \R_{\geq 0}$ be the $k$-DPP defined by $L.$ There exists an algorithm that with probability $\ge 1-\delta$, exactly samples from $\mu_k$ in $\Otilde{\sqrt{k}}$ parallel time using $\poly(n) \cdot  \log \frac k \delta$ machines.
    \item Let  $\mu: 2^{[n]} \to \R_{\geq 0}$ be the DPP defined by $L.$  There exists an algorithm, that with probability $\ge 1-\delta$, exactly samples from $\mu$ in $\Otilde{\sqrt{n}}$ parallel time using $\poly(n) \cdot \log \frac n \delta$ machines.
\end{enumerate}
\end{theorem} 

\Tag{We are also able to refine our results about DPPs so that the runtime is expressed in terms of typical sizes of the sets $S$ in the support, as measured by eigenvalues or traces of the matrix $L$. We leave the details to \cref{sec:symdppbound}, where we prove \cref{thm:sym DPP with bounded eigenvalue}.}

Finally, another distribution whose partition function can be computed via determinants is uniform perfect matchings in planar graphs. These distributions, alongside determinantal point processes, were raised as challenge distributions for parallel sampling \cite{anari2020sampling}. Unfortunately, planar perfect matchings are not entropically independent, and we cannot apply \cref{thm:main} to them. Nevertheless, using alternative techniques based on graph separators, we manage to obtain a similar quadratic speedup in sampling from them.
\begin{theorem}\label{thm:planar}
Let graph $G = (V,E)$ be planar and $ \mu$ be the uniform distribution over perfect matchings of $G.$ There exists an algorithm that samples from $\mu$ in $\Otilde{\sqrt{n}}$ parallel time using $\poly(n)$ machines.
\end{theorem}

\subsection{Techniques and algorithms}

Throughout, we heavily use the fact that for all distributions $\mu$ that we study in this paper, the marginals $\P_{S\sim \mu}{i\in S}$ can be computed in \Class{NC}, and that the distributions $\mu$ are self-reducible---by conditioning on element inclusion, we obtain another distribution in the same family of DPP variants or planar perfect matchings on a smaller graph. These two properties alone form the basis of the most classic, inherently sequential, algorithm for sampling described below.
\begin{Algorithm*}
\For{$i=1,\ldots,k$}{
    Compute the marginals of $\mu$ conditioned on elements $x_1,\dots, x_{i-1}$.\;
    Sample an element outside $x_1,\dots,x_{i-1}$ with probability proportional to the computed marginals. Call the sampled element $x_i$.\;
}
\Return{$\set{x_1,\dots,x_k}.$}\;
\end{Algorithm*}

Our main idea is to use rejection sampling to batch the iterations of this algorithm. Roughly speaking, we compute marginals of $\mu$ and sample a batch of elements $x_1,\dots, x_\l$ i.i.d.\ from these marginals. We then use rejection sampling to accept or reject the batch to make sure any set $\set{x_1,\dots,x_\l}$ is selected with probability given by the $\l$-order marginals $\propto \P_{S\sim \mu}{\{x_1,\dots,x_\l\}\subseteq S}$. Once we have a batch of elements successfully accepted, we continue sampling the next batch from the distribution conditioned on including this batch. A high-level description of this algorithm can be seen in \cref{alg:batched sample}. Our innovation is to implement the batch sampling step highlighted via (*) by i.i.d.\ sampling from marginals and performing a correction based on rejection sampling.

\begin{Algorithm}[ht!]
        \KwIn{$\mu: \binom{[n]}{k}\to \R_{\geq 0} $}
        $k_0 \leftarrow k$\;
        $\mu^{(0)}\leftarrow \mu$\;
		\For{$i=0,1,\ldots, 2\sqrt{k}$}{
			(*): Sample $ T_i\sim \mu^{(i)}$ with $\abs{T_i} = \ceil{\sqrt{k_i}}$\;
			Update $\mu^{(i+1)} \leftarrow \mu^{(i)} (\cdot  \mid T_i)$\;
			Update $k_{i+1} = k_i- \abs{T_i}$\;
		}
		\Return{$T:=\bigcup_i T_i.$}
		\caption{Batched sampling} \label{alg:batched sample}
	\end{Algorithm}
	
The sizes of the batches we sample dictate the parallel runtime of this algorithm. Even for symmetric DPPs, there is a natural barrier at batch size $\l \simeq \sqrt{k}$. Consider $L$ to be the Gram matrix of vectors $\set{e_1,e_1,\dots, e_k, e_k}$ (where every standard basis vector $e_i \in \R^n$ is repeated twice, so $k = \frac n 2$). The marginals of this DPP are uniform, but because of the Birthday Paradox, any sample of $\gg \sqrt{k}$ elements contains a pair of identical vectors with high probability, resulting in a joint marginal of $0$, i.e., an overwhelming probability of rejection. Hence, we must set the batch size $\l \lesssim \sqrt{k}$ to have a non-negligible acceptance probability.

\Tag{A significant difficulty that arises for DPP variants beyond symmetric DPPs is the lack of negative dependence. Roughly speaking, due to the lack of negative dependence, the acceptance probabilities used in symmetric DPPs for rejection sampling have to be scaled down in other cases by a factor of $\simeq 2^\l$; otherwise, we would sometimes have to accept with probability $>1$. See \cref{sec:hardex} for a detailed example of where this phenomenon can be observed. This scaling of $2^\l$ means the overall acceptance probability can be at most $2^{-\l}$, which forces the number of machines we use to sample, in parallel, possible batches for one iteration to scale up by a factor of $2^\l$. This limits $\l$ to be only logarithmic, and the parallel runtime would not be improved beyond logarithmic factors.}

\Tag<sigconf>{Unfortunately, beyond symmetric DPPs, rejection sampling only works up to batches of size $\Otilde{1}$. An example demonstrating this can be found in the full version of this paper.}
We overcome this challenge by replacing rejection sampling with approximate rejection sampling, where we allow the acceptance probabilities to go above $1$ on a small subset of the event space, and if we see any batch of this kind we declare the algorithm has failed. Our main insight is that such bad batches of elements must consist of large groups of highly correlated elements. On the other hand, we quantify limits on correlations in our distributions of interest by obtaining novel consequences of entropic independence \cite{AJKPV21}. Intuitively, we prove that correlations in any entropically independent distribution must be limited to small groups of elements, and a batch of $\simeq k^{\frac{1}{2}-c}$ elements will likely contain no more than one element from each correlated group. Formalizing this intuition, we prove \cref{thm:nsdpp,thm:pcdpp} in \cref{sec:ei}. \Tag{We give an example showing this sub-polynomial overhead in the parallel depth may be necessary for our batched rejection sampling approach in \cref{sec:hardex}.}

Finally, for planar perfect matchings, we use a completely different approach. We use planar separators to break the graph into smaller pieces; the pieces can then be processed in parallel, once we fix the part of the perfect matching touching the separator. 

\subsection{Further related work}

The works \cite{teng1995independent, anari2020sampling} study the problems of parallel sampling spanning trees and arborescences from graphs. While spanning trees are special cases of DPPs, these works parallelize an algorithm of \textcite{Aldous90,Broder89}, which is very specific to spanning trees, and does not have a known counterpart for general DPPs or other distributions studied in this work.

There have been several approaches to parallel sampling in the literature. Notably, \textcite{FHY21}, and more recently \textcite{liu2021simple}, showed how to efficiently parallelize a popular class of sampling algorithms based on Markov chains and obtain nearly optimal parallelism for several graphical models such as the hardcore, Ising, and proper coloring models in certain regimes. While these results manage to parallelize certain types of Markov chains, they do not apply to distributions studied in this work. A prerequisite for these parallelization techniques to work is the existence of a Markov chain with single-site updates, that is changing one coordinate in each move assuming the distribution is supported on a product space, that mixes in nearly-linear time. No such Markov chain is known for our application distributions.

Finally, subsequent to our results, \Tag<sigconf>{it was recently communicated to us that} \cite{parallel-sl} have obtained \Class{RNC} sampling algorithms for distributions that satisfy a stronger condition than entropic independence and using a stronger oracle than the simple counting oracle: an oracle that can compute for any given $\lambda\in \R_{>0}^n$, the value $\sum_{S}\mu(S)\prod_{i\in S}\lambda_i$. While their general result is not comparable to our main result, \cref{thm:main}, it does as a corollary give an \Class{RNC} sampling algorithm for Partition-DPPs. For nonsymmetric DPPs and planar perfect matchings, as far as we know, \Class{RNC} sampling is still open.

\Tag{
	\subsection*{Acknowledgment}
	\input{acks}
}

%% file: acks.tex
This work was supported by an NSF CAREER Award CCF-2045354 and a Sloan Research Fellowship. Thuy-Duong Vuong was supported by a Microsoft Research Ph.D.\ Fellowship. Callum Burgess was supported by the Stanford CURIS program.

%% file: overview.tex
\section{Overview}

In this section, we give a technical outline of our proofs and a roadmap for the rest of the paper.

\textit{Symmetric DPPs.} In \cref{sec:symmetric-dpp}, we prove our most basic result, \cref{thm:sym DPP} (sampling symmetric DPPs), as an introduction to our techniques and to demonstrate how we apply \cref{alg:batched sample}. Conveniently, symmetric DPPs exhibit strong negative dependence properties which the rest of our distributions do not. To prove \cref{thm:sym DPP}, we directly bound the acceptance probability of \cref{alg:batched sample} for batch size $\ell \simeq \sqrt{k}$. We show that directly applying negative dependence bounds the acceptance probability by $\exp(-\frac{\ell^2} k)$. \Tag{In \cref{sec:symdppbound} we obtain refined results for DPPs whose size is not constrained (i.e., not $k$-DPPs), but whose kernel matrix satisfies various forms of spectral boundedness, all related to the typical size of sampled sets.}

\textit{Entropic independence.} In \cref{sec:ei}, we provide a meta-result (\cref{thm:mainei}, a formal restatement of \cref{thm:main}) used to derive \cref{thm:nsdpp,thm:pcdpp} as corollaries. \Cref{thm:mainei} shows that for \emph{any} entropically independent distribution over subsets of size $k$, we can reduce sampling to marginal computations with  $\Otilde{k^{\frac{1}{2} + c}}$ parallel depth overhead, with a high probability of success. Here, $c$ is any constant, parameterizing the (polynomial) number of machines used.

To prove \cref{thm:mainei}, we use the entropic independence property to derive concentration bounds on the acceptance probability in \cref{alg:modified rejection sampling}.  As a first step towards this goal, we use entropic independence to demonstrate that up to parallel depth $\ell \approx k^{\frac{1}{2} - c}$, the Kullback-Leibler (KL) divergence between our target distribution (the order-$\ell$-marginals) and our proposal distribution (the product distribution on $1$-marginals) is bounded. This KL divergence bound does not suffice for our overall scheme; intuitively, it provides an ``average case'' bound on the log-acceptance probability of \cref{alg:modified rejection sampling}, whereas we would like to have a high probability bound since we need to union bound over at least $\sqrt{k}$ stages of rejection sampling. To simplify our concentration argument, we begin by assuming w.l.o.g.\ that our distribution has roughly uniform $1$-marginals, by using a subdivision process, similar to the ones used in \cite{AD20, ADVY21}. We then use comparison inequalities between KL divergences and (exponentiated) Renyi divergences for nearly-uniform distributions to bound moments associated with our rejection sampler's acceptance probabilities. Finally, we use these moment bounds to show that over a high-probability set of outcomes (in the sense of \cref{alg:modified rejection sampling}), the log of acceptance probability is a submartingale, yielding concentration via Markov's inequality.

\Tag{
\textit{Hard instance.} It is natural to ask: Can we improve the small additional overhead of $(k/\epsilon)^c$ in \cref{thm:mainei} (and hence, \cref{thm:nsdpp,thm:pcdpp}) to a smaller overhead, e.g.\ polylogarithmic? In \cref{sec:hardex}, we give a hard example showing that the subpolynomial overhead may be inherent to rejection sampling strategies, at least in the full generality of entropically independent distributions.
}

\textit{Planar graph perfect matchings.} In \cref{sec:planar}, we consider the problem of parallel sampling of perfect matchings from a planar graph (whose counts can be written as a determinant \cite{Kas67}). Sampling planar perfect matchings was raised as a challenge in \cite{anari2020sampling}; while our approach in \cref{sec:planar} departs from \cref{thm:mainei}, it still results in a quadratic speedup over na\"ive sequential sampling. We attain this speedup, \cref{thm:planar}, by leveraging parallel implementations of the planar separator theorem. \Tag{By sequentially sampling the portion of the matching adjacent to the vertices in a planar separator of size $O(\sqrt n)$, and recursing on the (geometrically smaller) disconnected components, we obtain a roughly-quadratic speedup for this sampling problem as well.}

%% file: prelim.tex
\section{Preliminaries}\label{sec:prelims}

In this section, we provide preliminaries for the rest of the paper.

We use $[n]$ to denote the set $\set{1,\dots,n}$. For a set $S$, $\binom{S}{k}$ denotes the family of subsets of size $k$. For a distribution $\mu: 2^{[n]} \to \R_{\geq 0}$ and $T \subseteq [n]$, define $\mu(\cdot \mid T)$ to be the distribution on $2 ^{[n] \setminus T} $ defined by $\mu(F \mid T) \propto \mu(F \cup T).$ We will sometimes use the shorthand $\mu_{\mid T}$.

For a measure or density function $\mu:\binom{[n]}{k}\to\R_{\geq 0}$, the generating polynomial of $\mu$ is the multivariate homogeneous polynomial defined as $g_\mu(z_1,\dots,z_n)=\sum_{S\in \binom{[n]}{k}} \mu(S)\prod_{i\in S} z_i$.

\subsection{Divergences}
Let $q$, $p$ be distributions over the same finite ground set $[n]$. We define the KL divergence and, for $\lam \ge 1$, the $\lam$-divergence between $q$ and $p$ as follows:
\begin{align*}\DKL{q \river p} &\coloneq \E*_q{\log \frac q p} = \sum_{i \in [n]} q_i \log \parens*{\frac{q_i}{p_i}}, \\
\D_\lambda{q \river p} &\coloneq \E*_p{\parens*{\frac q p}^\lambda} = \sum_{i \in [n]} q_i^\lambda p_i^{1 - \lambda}.
\end{align*}
We remark that our definition of $\D_\lambda{}$ is (up to a constant scalar multiplication) the exponential of the standard Renyi divergence of order $\lambda$. The KL and Renyi divergences exhibit the following useful bound\Tag{ whose proof we include for completeness in \cref{sec:proofs}}.
\begin{lemma}\label{lem:klrenyi}
Let $q, p$ be distributions over $[n].$ Suppose for some $ C \geq 1$ and $S \subseteq [n]$: $p_i \leq \frac C n$ for all $i \in [n]$ and $p_i \geq \frac{1}{C n}$ for all $i\in S.$
Then, for any $\lam \ge 1$, if $S=[n]$
\[\D_{\lam} {q \river p} \leq C^{\lam-1} \parens*{1+ n^{\lam-1} \lam(\lam-1) (\DKL{q \river p}  + \log C)}. \]
More generally,
\[\sum_{i\in S} q_i \parens*{\frac{q_i}{p_i}}^{\lam-1} \leq  C^{\lam-1} \parens*{1+ n^{\lam-1} \lam(\lam-1) (\DKL{q \river p}  + \log C)}. \]
\end{lemma}

\subsection{Determinantal point processes}\label{ssec:dppdef}
A DPP on $n$ items defines a probability distribution over subsets $Y \subseteq [n].$ It is parameterized by a matrix $L \in \R^{
n\times n}$: $\P_{L}{Y } \propto \det(L_Y ),$ where $L_Y$ is the principal submatrix whose columns and rows are indexed by $Y.$ We call $L$ the ensemble matrix. 
We define the marginal kernel $K$ of $\mathbb{P}_{L}$ by
\begin{equation} \label{eq:L to K}
    K = L (I+L)^{-1} = I - (I + L)^{-1} = (L^{-1} + I)^{-1}. 
\end{equation}
Then, $\det(K_A) = \P_L{A \subseteq Y}$ (a proof can be found in \cite{KuleszaT12}). This also implies $K \preceq I$ for symmetric $K.$ Conversely,
\begin{equation} \label{eq:K to L}
    L = K(I- K)^{-1} 
    = (I-K)^{-1}-I = (K^{-1}- I)^{-1}.
\end{equation}

Given a cardinality constraint $k$, the $k$-DPP parameterized by $L$ is a distribution over subsets $Y$ of size $k$, defined by $\mathbb{P}_{L}^k[Y] =\frac{ \det(L_Y) } {\sum_{\abs{Y'} =k } \det (L_{Y'})} . $ To ensure that $\mathbb{P}_L$ defines a probability distribution, all principal minors of $L$ must be non-negative: $\det(L_S ) \geq 0$. Any nonsymmetric (or symmetric) PSD matrix automatically has nonnegative principal minors \citep[Lemma 1]{Gartrell2019LearningND}.

Consider a matrix $L \in \R^{n\times n}$, partition $V_1\cup \dots \cup V_r = [n]$ of $[n],$ and tuple $\{c_i\}_{i=1}^r$ of integers.
The DPP with partition constraint (Partition-DPP) $\mu_{L;V,c}: 2^{[n]} \to \R_{\geq 0}$ is  defined by
\[\mu_{L; V,c}(S) \propto\1[\forall i: \abs{S \cap V_i} = c_i] \det(L_{S,S})\]

For any $Y \subseteq [n]$, if we condition the distribution $\mathbb{P}_{L}$ ($\mathbb{P}_{L}^k$ resp.) on the event that items in $Y$ are included in the sample, we still get a DPP ($(k-\abs{Y})$-DPP resp.); the new ensemble matrix is given by the Schur complement $L^Y = L_{\tilde{Y}} - L_{\tilde{Y}, Y} L_{Y, Y}^{-1} L_{Y, \tilde{Y}} $ where $\tilde{Y} = [n] \setminus Y.$ 

For Partition-DPPs, a similar statement holds. Conditioning $\mu_{L;V,c}$ on $Y$ being included in the set results in a Partition-DPP $\mu_{L^Y;V',c'}$ with ensemble matrix $L^Y$ and partition $V'_1 \cup \dots \cup V'_r = [n] \setminus Y$ with $V'_i = V_i \setminus Y,$ and $ c'_i = c_i - \abs{V_i \cap Y}$.

\Tag{We defer the proof of the following computational facts about DPPs to \cref{sec:proofs}.}
\begin{proposition}\label{prop:fast compute marginal}
 Suppose $\mu$ is one of the following distributions.
\begin{enumerate}
    \item k-DPP: $\mu(S) \propto \1[\abs{Y} = k] \det(L_S)$.
    \item DPP: $\mu(S) \propto \det(L_S)$.
    \item Partition-DPP: $\mu_{V,c}(S) \propto\1[\forall i\in [r]: \abs{S \cap V_i} = c_i] \det(L_{S})$ with $r = O(1)$.
\end{enumerate}
There are algorithms that perform the following tasks in $\Otilde{1}$-parallel time using $\poly(n)$ machines.
\begin{enumerate}
    \item Given $S\subseteq T\subseteq [n],$  exactly compute $\P_{X\sim \mu}{T\subseteq X \given S\subseteq X}$.
    \item Given $S \subseteq [n]$ and $t\in [n],$ exactly compute $\P_{T \sim \mu} {\abs{T} = t}.$
\end{enumerate}
\end{proposition}

\Tag{
Well-known concentration inequalities on DPPs imply that the size of a typical set is tightly concentrated around its mean. We defer the proof of the following to \cref{sec:proofs}.
\begin{lemma} \label{cor:concentration of size}
Let $\mu: 2^{[n]} \to \R_{\geq 0}$ be a strongly Rayleigh distribution. Suppose $\E_{S\sim \mu} {\abs{S}}\leq \sqrt{n} .$
 For $\epsilon \in (0,\frac 1 4) ,$ there exists an absolute constant $c > 0$ such that
 \[\P*_{S\sim \mu}{\abs{S} \geq c\sqrt{n\log \frac{1}{\epsilon} }} \leq \epsilon \]
 and
 \[\P*_{S\sim \mu}{\abs{S} \geq c\E_{S\sim \mu} {\card{S}} \log \frac{1}{\epsilon} } \leq \epsilon \]
\end{lemma}
}

\Tag{
\begin{remark} \label{cor:reduce DPP to k-DPP}
To sample from a DPP $\mu: 2^{[n]} \to \R_{\geq 0},$ we can first in constant parallel-time compute the distribution $\mathcal{H}$ on $[n]$ defined by $\P_{\mathcal{H}}{k} :=\P_{S \sim \mu} {\abs{S} = k}$ and sample the cardinality of the set $k $ from $\mathcal{H}$, and then sample $S$ from $\mu_k$ using the results of this paper. 
\end{remark}
}

%
%

Determinantal point processes (the symmetric kind) and their conditionings belong to a class of probability distributions called strongly Rayleigh \cite[see][ for definition]{BBL09}. Strongly Rayleigh distributions satisfy a useful property called negative correlation.
\begin{lemma}[Negative correlation] \label{lem:real stable negative correlation}
Suppose $\mu: 2^{[n]} \to \R_{\geq 0}$ is strongly Rayleigh. For any set $T$,
$\P_{S \sim \mu}{T \subseteq S} \leq  \prod_{i\in T} \P_{S\sim \mu} {i\in S}.$
\end{lemma}
\begin{lemma} \label{lem:dpp real stable}
 Let $L$ be a symmetric PSD matrix. The following distributions are strongly Rayleigh.
 
\begin{enumerate}
    \item $\mu_k: \binom{[n]}{k} \to \R_{\geq 0}$, the $k$-DPP defined by $L.$  
    \item $\mu: 2^{[n]} \to \R_{\geq 0}$, the DPP defined by $L.$ 
\end{enumerate}
\end{lemma}

\Tag{
\begin{corollary} \label{cor:hadamard ineq}
Let $K \in \R^{n \times n}$ satisfy $0\preceq K \preceq I$. For any set $T \subseteq [n]$, $\det(K_T) \leq \prod_{i\in T} K_{i,i}$.
\end{corollary}
\begin{proof}
Consider the DPP $\mu$ with kernel matrix $K.$ Apply \cref{lem:real stable negative correlation} and note that  $\det(K_T) = \P_{S\sim \mu}{T \subseteq S}$ and $K_{i,i} = \P_{S \sim \mu}{i\in S}.$
\end{proof}
}

\subsection{Entropic independence}
We recall the notions of fractional log-concavity \cite{AASV21} and entropic independence \cite{AJKPV21}.
\begin{definition}[{\cite{AASV21}}]\label{def:fractional-log-concavity}
	A probability distribution $\mu: \binom{[n]}{k} \to \mathbb{R}_{\geq 0}$ is $\alpha$\emph{-fractionally-log-concave} if $\log g_\mu(z_1^{\alpha}, \ldots, z_n^{\alpha})$ is concave over $\R_{>0}^n.$
	For $\alpha = 1$, we say $\mu$ is log-concave.
\end{definition}

To define entropic independence we need the definition of the ``down'' operator. In brief, $D_{k \to \ell}$ transitions from a set $S$ of size $k$ to a uniformly random subset of size $\l$.

\begin{definition}[Down operator]
	For $\l\leq k$ define the row-stochastic matrix $D_{k\to \l}\in \R_{\geq 0}^{\binom{[n]}{k}\times \binom{[n]}{\l}}$ by
	\[ D_{k\to \l}(S, T)=\1[T\subseteq S]\cdot \frac{1}{\binom{k}{\l}}.\]
\end{definition}

Note that for a distribution $\mu$ on size-$k$ sets, $\mu D_{k\to \l}$ will be a distribution on size-$\l$ sets. In particular, $\mu D_{k\to 1}$ will be the vector of normalized marginals of $\mu$: $\{\frac 1 k\P{i\in S}\}_{i\in [n]}$. We will use the following shorthand:
\begin{definition}
	For $\mu:\binom{[n]}{k}\to\R_{\geq 0}$, we use $\mu_\l$ to denote $\mu D_{k\to \l}$.
\end{definition}

\begin{definition}
    \label{def:entropic-independence}
A probability distribution $\mu$ on $\binom{[n]}{k}$ is said to be $\frac 1 \alpha$-entropically independent, for $\alpha \in (0,1]$, if for all probability distributions $\nu$ on $\binom{[n]}{k}$,
\[ \DKL{\nu_1\river\mu_1}=\DKL{\nu D_{k\to 1} \river \mu D_{k\to 1}}\leq \frac{1}{\alpha k}\DKL{\nu \river \mu}.  \]
\end{definition}

\begin{lemma}[\cite{AJKPV21}, Theorem 4] \label{lem:FLC to entropic-independence}
If $\mu$ is $\alpha$-FLC then $\mu$ and all conditional distributions of $\mu$, i.e.\ $\mu( \cdot \mid S)$ for any $S \subseteq [n]$, are $\frac 1 \alpha$-entropically independent.
\end{lemma}

\begin{lemma}[\cite{AASV21}] \label{lem:alpha FLC}
The following distributions are $\alpha$-FLC for $\alpha = \Omega(1)$ and $k \in [n]$.
\begin{enumerate}
    \item $k$-DPP and DPP defined by nonsymmetric PSD $L \in \R^{n \times n}$.
    \item Partition-DPP defined by symmetric PSD $L \in \R^{n \times n}$ and a partition $\{V_i\}_{i=1}^r$ with $r = O(1).$
\end{enumerate}
\end{lemma}

\subsection{Rejection sampling}
Consider distributions $\mu$, $\nu$ over the same domain, and parameter $C$ such that $\max\set*{\frac{\mu(x)}{\nu(x)}}\leq C.$ Assuming sample access to $\nu$, we can also sample from $\mu$ via rejection sampling as in \cref{alg:rejection sampling}.

\begin{Algorithm}[ht!]
\KwIn{parameter $C> 0$ such that $\max\set*{\frac{\mu(x)}{\nu(x)}}\leq C$}
Sample  $x\sim \nu.$\;
Accept and output $x$ with probability $ \frac{\mu(x)}{C \nu(x)}.$\;
\caption{Rejection sampling} \label{alg:rejection sampling}
\end{Algorithm}
When \cref{alg:rejection sampling} succeeds, its output distribution is exactly $\mu$.
\Cref{alg:rejection sampling} succeeds with probability \Tag<sigconf>{$1/C$.}
\Tag{\[\P{\text{accept}} = \sum_x \frac{\mu(x)}{C\nu(x)}\nu(x) = \frac{1}{C}. \]}
For any $\delta \in (0,1),$ by running $C \log\delta^{-1}$ copies of the algorithm in parallel and taking the first accepted copy, we can boost the acceptance rate to $1-\delta$, stated formally in the following.
\begin{proposition} \label{prop:rejection sampling}
There is an algorithm that with probability $1-\delta$, outputs a sample from $\mu$ in the same asymptotic parallel time as required to sample from $\nu$, using $O(C \poly(n)\log \frac 1 \delta )$ machines.
\end{proposition}

We consider the following modification of \cref{alg:rejection sampling} when we have a weaker assumption that for some $\Omega \subseteq \supp(\nu)$ we have $\max\set*{\frac{\mu(x)}{\nu(x)}\given x\in \Omega}\leq C$ and $\sum_{x\in \Omega}\mu(x)\geq 1-\epsilon$ for some $\epsilon \in [0,1).$ 
\begin{Algorithm}
\KwIn{$\Omega \subseteq \supp(\nu)$, parameter $C> 0$ such that $\max\set*{\frac{\mu(x)}{\nu(x)}\given x\in \Omega}\leq C$}
Sample  $x\sim \nu.$\;
If $x \in \Omega$, accept and output $x$ with probability $ \frac{\mu(x)}{C \mu(x)}.$\;
\caption{Modified rejection sampling} \label{alg:modified rejection sampling}
\end{Algorithm}

The guarantees of \cref{alg:modified rejection sampling} follow immediately from \cref{prop:modified rejection sampling} and that the restriction of $\mu$ to $\Omega$ has total variation distance at most $\eps$ from $\mu$. In particular, we will use \cref{prop:modified rejection sampling} with $\delta = \eps$ and output an arbitrary sample when it fails to accept a sample.

\begin{proposition}  \label{prop:modified rejection sampling}
There is an algorithm that outputs a sample from some $\tilde{\mu}$ in the same asymptotic parallel time as required to sample from $\nu$, using $O(C \poly(n) \log \frac 1 \eps )$ machines, where $\dTV{\tilde{\mu}, \mu} = O(\eps).$ 
\end{proposition}

%% file: symmetric.tex
\section{Symmetric DPP}\label{sec:symmetric-dpp}
Here, we prove our basic result, \cref{thm:sym DPP}. \Tag{We provide a strengthening for DPPs satisfying nontrivial spectral bounds, \cref{thm:sym DPP with bounded eigenvalue}, in \cref{sec:symdppbound}.} We first state helper bounds used in the proof.

\begin{lemma} \label{lem:sym DPP rejection sampling}
	Suppose $\mu$ on $\binom{[n]}{k}$ is negatively correlated. Let $\mu_t = \mu D_{k \to t}$ and $p_i = \P_{\mu}{i\in S}.$ Then $\mu_t (T)/\parens*{t! \prod_{i\in T} \frac{p_i} k} \leq \exp\parens*{\frac{t^2} k}$.
\end{lemma}

\begin{proof}
Note that $\mu_t(T) =$
\[ \binom{k}{t}^{-1} \P_{S\sim \mu}{T\subseteq S} = \frac{t!}{k(k-1) \cdots (k-t+1)} \P_{S\sim \mu}{T\subseteq S}.\]
Thus, by negative correlation, 
\begin{align*}
    \frac{\mu_t(T) }{ t!\prod_{i\in T} \frac{p_i} k } &= \frac{k^t }{k(k-1) \cdots (k-t+1)} \frac{\P_{S\sim \mu}{T\subseteq S}}{\prod_{i\in T} \P_{S\sim \mu}{i\in S}  }\\
    &\le \parens*{\prod_{i=1}^{t-1} \parens*{1-\frac{i}{k}}}^{-1} \leq \exp\parens*{\frac{t^2}{k}}
\end{align*}
where we used the facts that $1- x \geq e ^{-2x} $ and $\prod_{i=1}^{t-1} \exp\parens*{\frac{2i}{k}} = \exp\parens*{\frac{t^2 - t}{k}}$.
\end{proof}

\begin{proposition}\label{cor:sqrtk}
If step (*) takes $O(\tau)$-parallel time, \cref{alg:batched sample} takes $ O(\sqrt{k} \cdot \tau)$-parallel time.
\end{proposition}
\begin{proof}
Note that
\[k_{i+1} = k_i -\lceil \sqrt{k_i} \rceil \leq k_i -\sqrt{k_i} \leq \parens*{\sqrt{k_i} -\frac{1}{2}}^2.\]
and thus $\sqrt{k_{i+1}} \leq \sqrt{k_i} - \frac{1}{2}.$ Hence, since $t \geq 2 \sqrt{k}$ implies $\sqrt{k_t} \leq  \sqrt{k_0} - \frac{t}{2} \leq 0$,
the algorithm terminates in $O(\sqrt{k})$ iterations, and takes $O(\sqrt{k})$ parallel time. 
\end{proof}
\begin{proof}[Proof of \cref{thm:sym DPP}]
We first consider the case of sampling $k$-DPPs.
By \cref{lem:dpp real stable}, when $\mu$ is a $k$-DPP defined by a symmetric PSD ensemble matrix $L$, $\mu$ and the conditionals of $\mu$ are real-stable. In particular, all $\mu^{(i)}$ as defined in \cref{alg:batched sample} are real-stable and hence strongly Rayleigh. 

Consider some loop $i$, and let $\mu \equiv \mu^{(i)}.$ We use \cref{lem:sym DPP rejection sampling}
to implement step (*). In particular, we first compute the marginals $p_i=\P_{\mu}{i\in S}$ in $\Otilde{1}$ parallel time. Next, let $\nu$ be the distribution over ordered tuples $(i_1, \dots, i_t)\in [n]^t $ with 
\[\nu(\{i_1, \dots, i_t\}) = \prod_{r=1}^t \frac{p_{i_r}} k.\]
We can identify $\mu_t$ with the distribution $\mu_t^*$ over $[n]^t$ where we denote $\mu_t^*(\{i_1, \dots, i_t\}) = \frac{\mu(\set*{i_1, \dots, i_t}) }{t!}$. Let $\delta' = \frac \delta {2\sqrt{k}}$, where \cref{alg:batched sample} takes at most $2\sqrt k$ iterations by \cref{cor:sqrtk}. We run the rejection sampling algorithm in \cref{prop:rejection sampling}, which succeeds with probability $1-\delta',$ to sample from $\mu_t^*$ given samples from $\nu$, with $C \leq \exp(\tfrac{t^2} k)= O(1) $ since $t = \lceil \sqrt{k} \rceil.$ Clearly obtaining a sample from $\mu_t^*$ yields a sample from $\mu_t$ by forgetting the ordering on the elements.

Hence, each iteration $i$ of \cref{alg:batched sample} takes $\Otilde{1}$ parallel time. By \cref{cor:sqrtk}, the algorithm takes $O(\sqrt{k})$-parallel time. By a union bound, the success probability is $\ge 1- 2\sqrt{k} \delta' = 1-\delta.$ The number of machines used is as in \cref{prop:rejection sampling}, $O(\poly(n) \log \frac{k}{\delta}).$ 

The result for DPPs immediately follows by first sampling the size $\card{S}$ of $S\sim \mu$ and then sampling from the appropriate $k$-DPP.
\end{proof}

%% file: ei.tex
\section{Entropically independent distributions}\label{sec:ei}

In this section, we prove the main result. We will use \cref{thm:mainei} to derive our samplers for various entropically independent distributions, namely \cref{thm:nsdpp,thm:pcdpp}, which immediately follow from combining \cref{lem:FLC to entropic-independence,lem:alpha FLC,thm:mainei}. 

\begin{theorem}\label{thm:mainei}
Let $\mu: \binom{[n]}{k} \to \R_{\geq 0}$ be such that all its conditional distributions are $\frac 1 \alpha$-entropically independent with $\alpha = \Omega(1).$  Suppose we can compute marginals $\P_{\mu}{i \given S}$ for $S\subseteq [n]$ and $i\not \in S$ in $\Otilde{1}$ parallel time. For any constant $c > 0$ and any $\eps \in (0, 1)$, there exists an algorithm to sample from a distribution within total variation distance $\eps$ of $\mu$ in $\Otilde*{\sqrt{k} \cdot \parens*{\frac k \eps}^{c}}$ parallel time using $(\frac n \eps)^{O(c^{-1})}$ machines.
\end{theorem}

\subsection{Isotropic transformation}
\label{ssec:isotropic}

We first reduce to the case of near-isotropic distributions. Similarly to \cite{AD20,ADVY21},  we say a distribution $\mu: \binom{[n]}{k} \to \R_{\geq 0}$ is isotropic if for all $i \in [n]$, the marginal $\P_{S \sim \mu}{i \in S}$ is $\frac{k}{n}$. Prior work \cite{AD20} introduced the following subdivision process transforming an arbitrary $\mu: \binom{[n]}{k} \to \R_{\ge 0} $ to a nearly-isotropic $\mu': \binom{U}{k} \to \R_{\ge 0}$, while preserving entropic independence.
 
\begin{definition} \label{def:isotropic-transformation}
Let $\mu: \binom{n}{k} \to \mathbb{R}_{\geq 0}$ be an arbitrary probability distribution and assume that we have access to the marginals $p_1,\dots, p_n$ of the distribution with $p_1+\dots+p_n=k$ and $p_i = \P_{S\sim \mu}{i\in S}$ for all $i$. For a parameter $\beta \in (0,1)$, let
$t_i:=\ceil{\frac{n}{\beta k}p_i}$. We create a new distribution out of $\mu$ as follows: for each $i \in [n]$, create $t_i$ copies of element $i$ and let the collection of all copies be the new ground set: $U \coloneq \bigcup_{i = 1}^{n} \{i^{(j)}\}_{j \in [t_i]}$. Define the following distribution $\mu^{\text{iso}}: \binom{U}{k} \to \mathbb{R}_{\geq 0}$:
\[
\mu^{\text{iso}}\parens*{\set*{i_1^{(j_1)}, \ldots, i_k^{(j_k)}}}:=\frac{\mu(\{i_1, \ldots, i_k\})}{t_1 \cdots t_k}.
\]
We call $\mu^{\text{iso}}$ the \emph{isotropic transformation} of $\mu$.
\end{definition}
Another way we can think of $\mu^{\text{iso}}$ is that to produce a sample from it, we can first generate a sample $\set{i_1, \ldots, i_k}$ from $\mu$, and then choose a copy $i_m^{(j_m)}$ for each element $i_m$ in the sample, uniformly at random. Subdivision preserves entropic independence.

\begin{proposition}[{\cite[Proposition 19]{ADVY21}}] \label{prop:entropic-ind-subdivision}
If the distribution $\mu$ is $\frac 1 \alpha$-entropically-independent, then so is $\muiso$.
\end{proposition}
The following useful properties of $\muiso$ generalize \cite[Proposition 24]{ADVY21}. We defer the proof to \Tag{\cref{sec:proofs}}\Tag<sigconf>{the full version}.
\begin{proposition}
\label{prop:near-isotropic} 
Let $\mu: \binom{n}{k} \to \mathbb{R}_{\geq 0}$, and let $\muiso: \binom{U}{k} \to \R_{\geq 0}$ be the subdivided distribution from \cref{def:isotropic-transformation} for some $\beta$. Let $C = 1+\sqrt \beta.$ The following hold for $\muiso$.
\begin{enumerate}
\item Marginal upper bound: For all $i^{(j)} \in U$, the marginal $\P_{S \sim \mu^{\text{iso}}}{i^{(j)} \in S}\leq  C \frac{k}{\abs{U}}$.
\item Marginal lower bound: If $p_i:= \P_{S \sim \mu}{i \in S} 
\geq \frac{\sqrt \beta k} n$, then for all $j \in [t_i]$,
$\P_{S \sim \muiso}{i^{(j)} \in S}\geq \frac{k}{C\abs{U}}.$
Furthermore, letting $R : = \set*{i^{(j)} \given p_i  \geq \frac{\sqrt \beta k}{n}, j \in [t_i] }$ 
then for any $\l \leq k$
\[\sum_{S \in \binom{R}{\l}} \muiso_\l (S) \geq 1 -  \sqrt{\beta} \l. \]
\item Bounded ground set size:
$n \beta^{-1} \leq \card{U} \leq n(1+\beta^{-1})$. 
\end{enumerate}
\end{proposition}

\begin{remark} \label{remark:conversion}
For any $\l\in [k],$ suppose algorithm $\mathcal{A}$ can sample from within total variation distance $\epsilon$ of $\mu^{\text{iso}}_\l$. Then $\mathcal{A}$ can also be used to sample from within total variation distance $\epsilon$ of $\mu_\l$ using the same amount of (parallel) time and machines. 
\end{remark}
\subsection{KL divergence bound}\label{ssec:klbound}

Throughout this section and \cref{ssec:highprob}, let $\ell \in [k].$ We begin by proving a bound on the KL divergence between conditional marginals from an observed set. We denote by $p$ the vector of marginals, that is $p_i=\P_{S\sim \mu}{i\in S}$.

\begin{lemma}\label{lem:klcondmarg}
Let $S \in \binom{[n]}{t}$ for $t \le \frac{1}{2} k$. Let $\mu_{t + 1 \mid S}: [n] \setminus S \to \R_{\ge 0}$ be the marginal distribution of elements in $S' \sim \mu_{t + 1}$ conditioned on $S \subset S'$, namely $\mu_{t + 1 \mid S} \coloneq \mu(\cdot \mid S) D_{(k - t) \to 1}$. Then,
\begin{equation}\label{eq:klcondmarg}\DKL*{\mu_{t + 1 \mid S} \river \frac 1 k p} \le \frac{2}{\alpha k} \log\parens*{\frac 1 {\P_{T \sim \mu}{S \subset T}}} + \frac{2t}{k}. \end{equation}
\end{lemma}
This bound follows from the entropic independence of $\mu$. We defer the proof to \Tag{\cref{sec:proofs}}\Tag<sigconf>{the full version}.

By averaging \cref{lem:klcondmarg} over $\mu_S$ (the conditional distribution of $T \sim \mu$ on $S \subset T$), we immediately obtain the following corollary.

\begin{corollary}\label{cor:klcondmarg}
Let $t \le \frac{1}{2} k$. Then following the notation of \cref{lem:klcondmarg},
\[\sum_{S \in \binom{[n]}{t}} \mu D_{k \to t}(S) \DKL*{\mu_{t + 1 \mid S} \river \frac 1 k p} \le \frac{2t}{k}\parens*{\frac 1 \alpha \log\parens*{\frac{2n}{k}} + 1}. \]
\end{corollary}
\begin{proof}
It suffices to apply \cref{lem:klcondmarg}, and the calculation
\begin{gather*}
\sum_{S \in \binom{[n]}{t}} \mu D_{k \to t}(S) \log\parens*{\frac{1}{\P_{T \sim \mu}{S \subset T}}}\\
= \sum_{S \in \binom{[n]}{t}} \mu D_{k \to t}(S) \log \parens*{\frac 1 {\mu D_{k \to t}(S) \binom k t}} \\
= \sum_{S \in \binom{[n]}{t}} \mu D_{k \to t}(S) \log \parens*{\frac 1 {\mu D_{k \to t}(S)}} + \log \frac 1 {\binom k t} \\
\le \log \frac{\binom n t}{\binom k t} \le t\log\parens*{\frac{2n} k}.
\end{gather*}
The first equality used $\P_{T \sim \mu}{S \subset T} = \mu D_{k \to t}(S) \binom k t$, and the last line used that the negative entropy of a distribution supported on $N$ elements is bounded by $\log N$. 
\end{proof}

Finally, we use \cref{cor:klcondmarg} to derive a KL divergence bound between the distributions $\mu_{\ell}$ and $\mu'_{\ell}$, respectively the target and proposal distributions encountered in our rejection sampling scheme.

\begin{lemma}\label{lem:klmumup}
Let $\mu'_{j}$ be the distribution of the set formed by $j$ independent draws from $\frac 1 k p$. Let $\ell \le \frac{1}{2} k$. Then,
\[\DKL*{\mu_\ell \river \mu'_\ell} \le \frac{\ell^2}{k}\parens*{\frac 1 \alpha \log\parens*{\frac{2n}{k}} + 1}. \]
\end{lemma}
\begin{proof}
For any $j \in [\ell]$, following the notation of \cref{lem:klcondmarg}, 
\begin{multline*}
\DKL{\mu_j \river \mu'_j} - \DKL{\mu_{j - 1} \river \mu'_{j - 1}} =\\ \sum_{S \in \binom{[n]}{j - 1}} \mu D_{k \to (j - 1)}(S) \DKL*{\mu_{j \mid S} \river \frac 1 k p}
\le\\ \frac{2(j - 1)}{k}\parens*{\frac 1 \alpha \log\parens*{\frac{2n}{k}} + 1}.
\end{multline*}
In the first line, we used the chain rule of KL divergence, and in the second line used \cref{cor:klcondmarg}. Finally, the conclusion follows by telescoping the above display for $1 \le j \le \ell$.
\end{proof}

\Cref{lem:klmumup} bounds the KL divergence between $\mu_\ell$ and $\mu'_\ell$, which can be thought of as an average log of acceptance probability for our rejection sampling scheme. For constant $\frac 1 \alpha$, this bound suggests that we can take $\ell \approx \sqrt{k}$ and obtain an efficient sampler for $\ell$-marginals; however, it is only an average bound, whereas we need a high-probability bound; this is because we have multiple steps of rejection sampling and we cannot afford to fail in one step with constant probability. We make this rigorous in \cref{ssec:highprob}, where we use the tools from this section to give concentration bounds on the acceptance probability of rejection sampling.

\subsection{Concentration of acceptance probability}\label{ssec:highprob}

In this section, assume that we have already performed the transformation in \cref{prop:near-isotropic} parameterized by some $\beta$, and obtained a distribution $\nu^{\text{iso}}: \binom{U}{k} \to \R_{\geq 0}$ and a set $R\subseteq U$ of elements with lower bounds on marginals as given by \cref{prop:near-isotropic}. Let $\nu:= \nu^{\text{iso}}$ and let $\nu'$ be defined analogously to \cref{ssec:klbound}.
Our goal is to sample from within $\epsilon$ total variation of $\nu_{\ell}$ for a suitably chosen $\ell,$ which also implies that we can sample from within $\epsilon$ of $\mu_{\ell}$ (see \cref{remark:conversion}). Our algorithm will be the modified rejection sampler (\cref{alg:modified rejection sampling}). 

To use \cref{alg:modified rejection sampling} with $P = \nu D_{k \to \ell}$ and $Q$ the $\ell$-wise product distribution drawing from $\frac 1 k p$, we first define a relevant high-probability set $\Omega$ on our state space $\xset \coloneq \binom{U}{k}$. Our set $\Omega$ will be a subset of the following set, for some $\varepsilon > 0$ we will choose later:
\[\tOmega_\varepsilon \coloneq \set*{S \in \binom{U}{\ell} \given \nu_{|T|}(T) \ge \varepsilon^{|T|},\; \forall T \subseteq S }. \]
In other words, $\tOmega_\varepsilon$ contains all sets $S$ such that all subsets $T \subset S$ are relatively well-represented according to $\nu_{|T|}(T)$. We begin with an observation lower bounding the measure of $\tOmega_\varepsilon$.

\begin{lemma}\label{lem:omegabound}
For any $0 \le \varepsilon \le \frac{1}{2\abs{U}\ell}$, we have
$\sum_{S \not\in \tOmega_\varepsilon} \nu_\ell(S) \le 2\abs{U}\ell\varepsilon$.
\end{lemma}
\begin{proof}
Let $\cert \coloneq \cup_{t = 1}^\ell \{T \in \binom{U}{t} \mid \nu_t(T) \le \varepsilon^t\}$. For any $S \not\in \tOmega_\varepsilon$, we say $T \in \cert$ is a ``certificate'' of $S$ if $T \subset S$; every $S \in \tOmega_\varepsilon^c$ has at least one certificate, so there is a map $\map: \tOmega_\varepsilon^c \to \cert$. Moreover, for some $T \in \cert$, let $\tOmega_\varepsilon^c(T)$ be the set of all $S \in \tOmega_\varepsilon^c$ such that $\map(S) = T$. Then since
\[\nu_t(T) = \sum_{S \supseteq T} \frac{1}{\binom \ell t} \nu_\ell(S) \implies \sum_{S \in \tOmega_\varepsilon^c(T)} \nu_\ell(S) \le \binom{\ell}{t}\nu_t(T),\]
summing over all $T \in \cert$ yields
\begin{align*}
\sum_{S \in \tOmega_\varepsilon^c} \nu_\ell(S) &\le \sum_{T \in \cert} \binom{\ell}{t}\nu_{|T|}(T) 
= \sum_{1 \le t \le \ell} \sum_{\substack{T \in \binom{U}{t} \\ \nu_t(T) \le \varepsilon^t}} \binom{\ell}{t}\nu_t(T) \\
&\le \sum_{1 \le t \le \ell} (\abs{U}\ell \varepsilon)^t \le \frac{\abs{U}\ell\varepsilon}{1 - \abs{U}\ell\varepsilon} \le 2\abs{U}\ell\varepsilon.
\end{align*}
The last line used the approximations $\binom{\ell}{t} \le \ell^t$, $\binom{\abs{U}}{t} \le \abs{U}^t$.
\end{proof}

For the remainder of the section we will specifically use $\varepsilon = \frac{\epsilon}{32\abs{U}\ell}$, such that $\tOmega_\varepsilon$ captures at least a $1 - \frac \epsilon {16}$ fraction of the mass of $\binom{U}{\ell}$ according to $\nu_\ell$. We will also drop $\varepsilon$ from $\tOmega_\varepsilon$ for simplicity. 

Our next goal is to show that almost all of the sets in $\tOmega$ have a polynomially bounded acceptance probability when the proposal is given by independent draws from $\frac 1 k p$. Consider iteratively building a set $S_t$ for all $1 \le t \le \ell$, where $S_t$ is a random variable formed by $S_{t - 1} \cup \{i_t\}$ for $i_t \sim \frac 1 k p$. In particular, we use $i_t$ to denote the $t^{\text{th}}$ draw from $\frac 1 k p$ in this process. For parameters $\tau, \gamma \ge 0$ to be defined later, iteratively define the random variables:
\begin{align*}
Y_{t + 1} &\coloneq Y_t \exp\parens*{\Delta_{t + 1}}, \\
\Delta_{t + 1} &\coloneq \begin{cases}
\gamma \log\parens*{\frac{\nu_{t + 1 \mid S_t} (i_{t + 1})}{ \frac 1 k p_{i_{t + 1}}}} - \tau & \nu_t(S_t) \ge \varepsilon^t \text{ and } i_{t+1} \in R \\
-\infty & \text{ otherwise}
\end{cases}
\end{align*}
with $C = 1 + \sqrt \beta$ and $R$ as defined in \cref{prop:near-isotropic}. Also by \cref{prop:near-isotropic}, $p_{i} \leq \frac{Ck}{\abs{U}}$ for all $i\in U.$ We use the convention $\exp(-\infty) = 0.$

We next prove that $Y_{t + 1}$ is a submartingale for appropriate parameter choices.

\begin{lemma}\label{lem:submartingale}
Let $S_t = T$ have $\nu_t(T) \ge \varepsilon^t.$ Assume that $\sqrt \beta \leq \min\set*{\frac{1}{3\gamma} ,\frac{t}{\alpha k} \log \frac 1 \varepsilon} $. Then,
\[\tau \ge \abs{U}^\gamma \gamma(1 + \gamma) \cdot \parens*{\frac{12t}{\alpha k} \log \frac 1 \varepsilon} \implies \E*_{i \sim \nu_{t + 1 \mid S_t}}{Y_{t + 1} \given S_t = T} \le Y_t. \]
\end{lemma}
\begin{proof}
If $ Y_t = 0$ then $Y_{t+1}=0$ by definition. In the following, assume $Y_t > 0.$
By the definition of $Y_{t + 1} = Y_t \exp(\Delta_{t + 1})$, and since $\nu_t(T) \ge \varepsilon^t$, we have
\begin{align*}
&\E*_{i \sim \nu_{t + 1 \mid T}}{\frac{Y_{t + 1}}{Y_t} \given S_t = T} \\
&= \exp\parens*{-\tau} \E*_{i \sim \nu_{t + 1 \mid T}}{\1_{i\in R}\parens*{\frac{\nu_{t + 1 \mid T} (i)}{ \frac 1 k p_{i}}}^\gamma \given S_t = T} \\
&= \exp\parens*{-\tau} \sum_{i\in R} \nu_{t + 1 \mid T} (i) \parens*{\frac{\nu_{t + 1 \mid T} (i)}{ \frac 1 k p_{i}}}^\gamma\\
&\le \exp\parens*{-\tau} C^{\gamma} \parens*{1 + \abs{U}^\gamma \gamma(1 + \gamma) \parens*{\DKL*{\nu_{t + 1 \mid T} \river \frac 1 k p} + \log C }} \\
&\le \exp\parens*{-\tau} (1 + 2 \sqrt \beta \gamma) \parens*{1 + \abs{U}^\gamma \gamma(1 + \gamma) \cdot \parens*{\frac{4t}{\alpha k} \log \frac 1 \varepsilon + \sqrt \beta} }.
\end{align*}
The second-to-last inequality uses \cref{lem:klrenyi} with $C = 1 + \sqrt \beta$, and the last inequality used \cref{lem:klcondmarg}, which shows that since $\nu_t(T) \ge \varepsilon^t,$
\begin{align*}\DKL*{\nu_{t + 1 \mid T} \river \frac 1 k p} &\le \frac{2}{\alpha k} \log\parens*{\frac{1}{\P_{S \sim \nu} {T \subset S}}} + \frac{2t} k \\
&= \frac{2}{\alpha k} \log\parens*{\frac{1}{\nu_t(T) \binom{k}{t}}} + \frac{2t} k \le \frac{4t}{\alpha k} \log \frac 1 \varepsilon,\end{align*}
as well as $\log(1+\sqrt \beta) \leq \sqrt \beta$ and \[  (1+ x)^{\gamma} \leq e^{ x\gamma} \leq 1+2 x \gamma\] for $x \gamma := \sqrt \beta \gamma \leq \frac 1 3.$
The conclusion follows from 
\begin{align*}
    &(1 + 2 \sqrt \beta \gamma) \parens*{1 + \abs{U}^\gamma \gamma(1 + \gamma) \cdot \parens*{\frac{4t}{\alpha k} \log \frac 1 \varepsilon + \sqrt \beta} }\\
    &\leq 1 + 2\sqrt \beta \gamma + \abs{U}^\gamma \gamma(1 + \gamma) \cdot \parens*{\frac{5t}{\alpha k} \log \frac 1 \varepsilon } (1+ 2\sqrt \beta \gamma) \\
    &\leq 1 + \gamma \parens*{\frac{2t}{\alpha k} \log \frac 1 \varepsilon } + \abs{U}^\gamma \gamma(1 + \gamma) \cdot \parens*{\frac{5t}{\alpha k} \log \frac 1 \varepsilon }\parens*{1 + \frac 2 3}\\
    &\leq 1 + \abs{U}^\gamma \gamma(1 + \gamma) \cdot \parens*{\frac{12t}{\alpha k} \log \frac 1 \varepsilon }  \leq 1+\tau \leq \exp(\tau).\qedhere
\end{align*}
\end{proof}

Now, applying \cref{lem:submartingale} with the definition of $\tOmega' \coloneq \tOmega\cap \set*{S \in \binom{R}{\ell} }$ allows us to obtain a high-probability bound on the acceptance probability of our rejection sampling scheme.

\begin{lemma}\label{lem:mostomegagood}
Let $B \ge 1$. For sufficiently small $\epsilon \in (0, 1)$, and $\ell \in [k]$ satisfying $12\ell^2 \parens{\frac{16} \epsilon}^{\frac{3}{B}} \log \frac 1 \varepsilon\le \alpha k$ and assuming our choices of parameters satisfy
$\sqrt \beta\leq \min\set*{\frac{1}{3\gamma} ,\frac{1}{\alpha k} \log \frac 1 \varepsilon}$, we have
\[\P*_{S_\ell \sim \nu_\ell}{\frac{\nu_\ell(S_\ell)}{\nu'_\ell(S_\ell)} \ge \abs{U}^B \given S_\ell \in \tOmega'} \le \frac \epsilon 8. \]
\end{lemma}
\begin{proof}
Throughout this proof, we will assume
\[\gamma = \frac{2\log \frac{16}{\epsilon}}{B\log \abs{U}}, \; \tau = \frac{\log \frac{16}{\epsilon}}{\ell}.\]
We first observe that our parameter choices indeed satisfy the condition on $\tau$ used in \cref{lem:submartingale}:
\begin{align*}
\frac{\gamma(1 + \gamma)\abs{U}^{\gamma}}{\log\parens*{\frac{16}{\epsilon}}} \le \parens*{\frac{16}{\epsilon}}^{\frac {3} B}
\implies \gamma(1 + \gamma)\abs{U}^{\gamma} \cdot \parens*{\frac{12\ell}{\alpha k} \log \frac 1 \varepsilon } \cdot \frac{\ell}{\log \frac{16} \epsilon} \le 1.
\end{align*}

In the following we denote $\hmu_j$ to be the joint distribution of $\{i_1, i_2, \ldots, i_j\}$ where $i_1 \sim \nu_1$, $i_2 \sim \nu_{2 \mid S_1 = \{i_1\}}$, and so on. In other words, if $S_\ell$ is the unordered set of $\{i_1, i_2, \ldots, i_\ell\}$, we have $\nu_\ell(S_\ell) = \ell! \cdot \hmu(\{i_1, i_2, \ldots, i_\ell\})$. We similarly define $\hmu'_j$ so that $\nu'_\ell(S_\ell) = \ell! \cdot \hmu'_\ell(\{i_1, i_2, \ldots, i_\ell\})$. For $S_\ell \in \binom{U}{\ell}$, and some realization $\{i_1, i_2, \ldots, i_\ell\}$ whose unordered set is $S_\ell$, cancelling a factor of $\ell!$ yields
\begin{equation}\label{eq:Ldef}L(S_\ell) \coloneq \frac{\nu_\ell(S_\ell)}{\nu'_\ell(S_\ell)} = \frac{\hmu_\ell(\{i_1, i_2, \ldots, i_\ell\})}{\hmu'_\ell(\{i_1, i_2, \ldots, i_\ell\})} = \prod_{j \in \ell} \frac{\nu_{j \mid S_{j - 1} = \{i_1, i_2, \ldots, i_{j - 1}\}}(\{i_j\})}{\frac 1 k p_{i_j}}\end{equation}
where $\nu'_\ell$ is the distribution of the unordered set corresponding to $\ell$ draws from $\frac 1 k p$.
Next, we apply \cref{lem:submartingale} which yields a submartingale property on $Y_\ell$. Letting $\1_{S_\ell \in \tOmega'}$ be the $0$-$1$ valued indicator function of the event $S_\ell \in \tOmega'$, we compute
\begin{align*}
1 = Y_0 &\ge \E*_{\{i_1, i_2, \ldots, i_\ell\} \sim \hmu_\ell}{Y_\ell \cdot \1_{S_\ell \in \tOmega'}} \\
&= \P*_{S_\ell \sim \nu_\ell}{S_\ell \in \tOmega'} \times \\
& \E_{\set{i_1, i_2, \ldots, i_\ell} \sim \hmu_\ell}{\exp( -\ell\tau + \gamma\log L(S_\ell)) \given S_\ell \in \tOmega'}.
\end{align*}

In the last two expressions, $S_\ell$ denotes the set $\set{i_1, i_2, \ldots, i_\ell}$. The first inequality used the fact that $Y_\ell$ is always nonnegative, and whenever $S_\ell \in \tOmega'$ we can apply \cref{lem:submartingale} to all subsets in the stages of its construction. The second line follows since whenever $S_\ell \in \tOmega'$, we are always in the first case in the definition of $\Delta_{t + 1}$, and then we can apply \eqref{eq:Ldef}. Hence, for any $B \ge 0$,
\begin{align*}
\P*_{S_\ell \sim \nu_\ell}{S_\ell \in \tOmega'} \cdot \E_{\{i_1, i_2, \ldots, i_\ell\} \sim \hmu_\ell}{\exp\parens*{\gamma\log L(S_\ell)} \given S_\ell \in \tOmega'} \\ \le \exp\parens*{\ell\tau}
\end{align*}
which implies
\begin{align*}
\P_{\set{i_1, i_2, \ldots, i_\ell} \sim \hmu_\ell}{\log L(S_\ell) \ge B\log \card{U} \given S_{\ell} \in \tOmega'} \\ \le 2 \exp\parens*{\ell \tau - \gamma B \log \card{U}},
\end{align*}
where we used Markov's inequality, and that $\tOmega'$ captures at least half the mass of $\nu_\ell$. However, every permutation giving rise to the unordered set $S_\ell$ is equally likely under $\hmu_\ell$, so by aggregating permutations, this can be rewritten as the desired 
\[\P*_{S_\ell \sim \nu_\ell}{L(S_\ell) \ge \abs{U}^B \given S_\ell \in \tOmega'} \le 2 \exp\parens*{\ell \tau - \gamma B \log \abs{U}} = \frac \epsilon 8.\qedhere \]
\end{proof}

Finally, we combine \cref{lem:omegabound,lem:mostomegagood} to prove our main result.

\begin{lemma}\label{lem:subroutine}
Let $B \ge 1$, and for a sufficiently small constant below, suppose
$\ell^2 = O\parens*{\frac{\alpha k}{\log \frac n \eps } \cdot \epsilon^{\frac 3 B}}$.
There is a parallel algorithm using $O( (nk^2\epsilon^{-2})^B \log \frac 1 \epsilon)$ machines which runs in $O(1)$ time and returns a draw from a distribution within total variation distance $\frac \epsilon 2$ of $\mu D_{k \to \ell}$.
\end{lemma}
\begin{proof}
Without loss of generality, we can assume $n$ is at least a sufficiently large constant, else the standard sequential sampler has parallel depth $\Otilde{1}$. 
We set 
\[\sqrt \beta :=\frac{\epsilon}{32 k}  \leq \min \set*{\frac{1 }{\alpha k} \log \frac{1}{\varepsilon},\frac 1 {3\gamma}} = \min \set*{\frac{1 }{\alpha k} \log \frac{1}{\varepsilon}, \frac{B \log n}{6 \log \frac{16}{\epsilon} }},\] which clearly satisfies the assumption of \cref{lem:submartingale} for sufficiently large $n$.
  Set $\varepsilon = \frac{\epsilon}{32 \abs{U} \ell}.$ Combining \cref{prop:near-isotropic,lem:omegabound} and using a union bound, we have 
\begin{align*}\nu_{\ell} (\tOmega') &\geq 1- (1-\nu_{\ell}(\tOmega)) - \parens*{1- \nu_{\ell}\parens*{\set*{S \in \binom{R}{\l} } }} \\&\geq 1- 2 \sqrt \beta k - 2\abs{U} \l \varepsilon  \geq 1 - \frac{\epsilon}{8}. \end{align*}
By \cref{prop:near-isotropic}, $\abs{U} \leq 2n\beta^{-1} = O(nk^2 \epsilon^{-2})$ and 
\[\log \frac{1}{\varepsilon} = O\parens*{\log \frac{ \abs{U}\l }{\epsilon}  }= O\parens*{\log \frac n \epsilon}\]
Thus, this setting of $\ell$ and $\beta$ satisfies the assumption of \cref{lem:mostomegagood}. Hence, the subset $\Omega \subset \tOmega'$ which satisfies the conclusion of \cref{lem:mostomegagood} has measure at least $1 - \frac \epsilon 4$ according to $\nu_\ell$.  Using \cref{alg:modified rejection sampling,remark:conversion}, we can sample from within total variation $\frac \epsilon 2$ from $\nu D_{k \to \ell}$ and $\mu  D_{k \to \ell}$ in $\Otilde{1}$-time using $O(\abs{U}^B \log\frac{1}{\varepsilon})$, which equals $O( (nk^2\epsilon^{-2})^B \log \frac 1 \epsilon)$ machines. We note that to implement our modified rejection sampling, it suffices to check that the likelihood ratio is bounded, which will certainly be the case for all elements in $\tOmega'$, and if there are other sets with bounded likelihood ratio this only improves the total variation distance guarantee.
\end{proof}

\subsection{Proof of \cref{thm:mainei}}\label{ssec:proofei}

Now we combine everything to prove \cref{thm:mainei}.

\begin{proof}[Proof of \cref{thm:mainei}]
Since we can always sample in $\Otilde{k}$ parallel time, the statement is nontrivial only for $c \leq \frac{1}{2}.$
Set $\epsilon' \gets \frac \epsilon k$. As in \cref{cor:sqrtk}, it suffices to repeatedly sample from $\mu D_{k \to \ell}$ for some choice of $\ell$ respecting the bound in \cref{lem:subroutine}, within total variation $\epsilon'$. We then condition on this set and then repeat. By the coupling characterization of the total variation distance, the resulting distribution will be at total variation $\epsilon$ from $\mu$, since this process will terminate within $k$ rounds. It is straightforward to see this will terminate in  $O\parens*{\sqrt{k/\parens*{\frac{\alpha }{\log \frac{n}{\eps'} }\cdot {\eps'}^{\frac 3 B} }} }$ iterations by a variation of the proof of \cref{cor:sqrtk} and the maximum allowable $\ell$ in \cref{lem:subroutine}. Setting $B = \frac{3}{c }$ gives the desired bound on the number of iterations.
\end{proof}

%% file: planar.tex
\section{Perfect matchings in planar graphs}\label{sec:planar}

In this section, we give our parallel algorithm for sampling perfect matchings from a planar graph $G = (V, E)$. Our algorithm departs somewhat from the rejection sampling-based framework used to prove \cref{thm:mainei} and its specializations, but we include it as it highlights a different strategy for attaining quadratic parallel speedups by using oracles, catered to the structure of planar graphs. At a high level, our algorithm recursively finds a \emph{planar separator}, which is a small set of vertices $S \subseteq V$ whose removal decomposes $V$ into two disconnected components $V_1$ and $V_2$ containing roughly the same number of vertices. We then use our counting oracle to match the vertices in $S$ sequentially, and then recursively solve the sampling problem in the subgraphs on $V_1$ and $V_2$.

\begin{proof}[Proof of \cref{thm:planar}]
We find a planar separator $S\subseteq V$ of size $O(\sqrt{n})$ in $\Otilde{1}$ parallel time using $\poly(n)$ machines, where removing $S$ from $G$ results in disconnected components $V_1, V_2$ containing at most $\frac {2n}{3}$ vertices each \cite{GM87}.
To sample a perfect matching recursively, we first sample matching endpoints of vertices in $S$ as follows. Label the vertices in $S$ by $v_1,\ldots, v_{\card{S}}.$ Let $v = v_1$, and compute in parallel the probability $p_{u}$ that a random matching contains the edge $(u,v)$.\footnote{To see that we can implement this oracle, \cite{Kas67} shows counting perfect matchings on a planar graph is reducible to determinant computation, which is parallelizable \cite{csanky1975fast}.} We then sample a matched edge for $v$ from this distribution, set $v$ to the next unmatched vertex in $S$, and continue conditioned on prior matched edges. Sequentially proceeding over $S$ allows us to sample over partial matchings containing all of $S$ in parallel time $\Otilde{\sqrt n}$. We then remove $S$ (and vertices matched to $S$ in this process) from $V_1$, $V_2$ (which can only decrease their sizes), and return the union of the matching of $S$ with the results of recursive calls to $V_1$ and $V_2$. 

The recursion depth is $\Otilde{1}$, and each round of recursion takes parallel time $O(|S|) = O(\sqrt n)$, proving the overall parallel depth. For the work bound, let $Q(n)$ be a bound on the number of machines used in the top level of recursion (excluding the recursive calls to $V_1$, $V_2$ after matching $S$), and let $P(n)$ be the total bound on the number of machines used; assume without loss $Q(n) = \Omega(n^2)$. Then,
\[P(n) = 2 P(2n/3) + Q(n) \implies P(n) = O(Q(n)).\qedhere \]

\end{proof}

%% file: hard.tex
\section{Hard instance for rejection sampling}\label{sec:hardex}

In this section, we give a simple hard instance of a fractionally log-concave distribution, which demonstrates that the dependence on $k$ in \cref{thm:mainei} may be inherent to our rejection sampling strategy. In particular, it is natural to hope that we can improve \cref{thm:mainei} to obtain a parallel depth of $\sqrt{k} \cdot \text{polylog}(k)$, as opposed to $k^{\frac{1}{2} + c}$ for a constant $c$. Here, we give an example that suggests that new algorithmic techniques may be necessary to obtain this improvement.

Our hard distribution $\mu: \binom{[n]}{k} \to \R_{\ge 0}$ will be defined as follows. Let $n$ and $k$ be even, and consider a partition of the ground set $[n]$ into pairs $S_i \coloneq (2i - 1, 2i)$ for all $i \in [\frac n 2]$. Then, the distribution $\mu$ is uniformly supported on sets of the form
\begin{equation}\label{eq:sdef}S \coloneq \bigcup_{i \in S'} S_i, \text{ where } S' \in \binom{\bracks*{\frac{n}{2}}}{\frac{k}{2}}. \end{equation}
In other words, $\mu$ randomly chooses $\frac k 2$ indices between $1$ and $\frac n 2$ and takes the $k$ elements formed by including the pairs corresponding to those indices. It is known that $\mu$ is $\Omega(1)$-FLC (see \cite{ADVY21}). To simplify notation, we will assume that $k = o(n)$ and $\ell = o(k)$. We will also assume there is a constant $B$ such that we have access to $n^B$ parallel machines. Following the guarantees of~\cref{alg:modified rejection sampling} in \cref{prop:modified rejection sampling}, if we are willing to tolerate a total variation distance of $\delta$ from $\mu_{\ell}$, we need to show that with probability at least $1 - \delta$, $S \sim \mu_\ell$ satisfies
\begin{equation}\label{eq:goodS}\frac{\mu_\ell(S)}{\mu'_{\ell}(S)} \le n^B.\end{equation}
Here and throughout the following discussion, $\mu'_{\ell}(S) = \frac{\ell!}{n^\ell}$ is the probability $S$ is formed by $\ell$ independent draws from the uniform distribution on $[n]$. In particular, the $1$-marginal distribution of $\mu$ is uniform, so this is the proposal distribution used by rejection sampling. 

Our argument on the tightness of our rejection sampling proceeds as follows. Say that a set $S \in \binom{[n]}{\ell}$ has $t$ ``duplicates'' if, amongst the elements of $S$, there are exactly $t$ pairs of elements belonging to the same $S_i$. For example, for $\ell = 4$ we say the set $\{1, 2, 3, 5\}$ contains $1$ duplicate, the pair $(1, 2)$. We first show that for a set $S$ to satisfy \eqref{eq:goodS}, it cannot contain more than $t = O(B)$ duplicates. We then show that this limitation, along with attaining a failure probability $\delta$ inverse-polynomial in $k$, forces us to choose $\ell = k^{\frac{1}{2} - c}$ for a constant $c > 0$ which may depend on $B$.

\paragraph{How many duplicates can we afford?} Suppose $S \in \binom{[n]}{\ell}$ contains $t$ duplicates. Each permutation of $S$ is equally likely to be observed by either of the following processes starting from $T_0 = \emptyset$ (we use $T_i$ to denote an ordered set, and $S_i$ to denote its unordered counterpart, for all $i \in [\ell]$).
\begin{enumerate}
    \item For $i \in [\ell]$, draw $j \in [n]$ uniformly at random and add it to $T_{i - 1}$ to form $T_i$.
    \item For $i \in [\ell]$, draw $j \in [n]$ according to the marginal distribution of $\mu_{\ell}$ conditioned on including $T_{i - 1}$ and add it to $T_{i - 1}$ to form $T_i$.
\end{enumerate}
Hence, to bound $\frac{\mu_\ell(S)}{\mu'_{\ell}(S)}$ as needed by \eqref{eq:goodS}, it suffices to fix a permutation $T_\ell$ of $S_\ell = S$ and bound the ratios of the probabilities $T_\ell$ is observed according to each of the above processes. It is observed with probability $n^{-\ell}$ according to the first process above, so satisfying \eqref{eq:goodS} means the probability $T_\ell$ is observed by the second must be at most $n^{B - \ell}$.

It is straightforward to see that the probability we observe each second element in a duplicate pair in the relevant round $i \in [\ell]$ is $\Theta\parens*{\frac 1 k}$. On the other hand, the probability of observing each singleton in its round is $\Theta\parens*{\frac 1 n}$. For $k = o(n)$, this shows that to meet \eqref{eq:goodS} we must have 
\[ \parens*{\Theta\parens*{\frac 1 n}}^{\ell - t} \parens*{\Theta\parens*{\frac 1 k}}^{t} \le n^{B - \ell}.\]
This shows that we must have $t$ at most a constant (depending on $B$). 

\paragraph{Probability of $t$ duplicates.} Let $t$ be a constant. Recall that the distribution $\mu$ is uniform over all sets of the form \eqref{eq:sdef}, and a sample from $\mu_\ell = D_{k \to \ell} \mu$ is formed by sampling a set $S \sim \mu$ and then randomly selecting one of the $\binom{k}{\ell}$ subsets of $S$. Hence, it suffices to fix some $S$ of the form \eqref{eq:sdef} and bound the probability that this downsampling process results in a subset with $t$ duplicates. By symmetry of $\mu$, we lose no generality by only considering the set $S = \cup_{i \in [\frac k 2]} S_i$.

Now, for a constant $t$, the number of subsets $S$ of size $\ell$ with exactly $t$ duplicates is
\[\binom{\frac k 2}{t} \cdot \binom{\frac k 2 - t}{\ell - 2t} \cdot 2^{\ell - 2t}.\]
The first term corresponds to choosing which $t$ sets $S_i$ will be fully included, the second corresponds to choosing which sets the remaining $\ell - 2t$ elements come from, and the third is because for each of the non-duplicated sets we have two options. Hence, the probability a draw from $\mu_\ell$ has exactly $t$ duplicates for constant $t$ scales as
\begin{align*}\frac{\binom{\frac k 2}{t} \cdot \binom{\frac k 2 - t}{\ell - 2t} \cdot 2^{\ell - 2t}}{\binom{k}{\ell}} &= \parens*{\Theta\parens*{\frac k \ell}}^{-\ell} \cdot \parens*{\Theta\parens*{k}}^t \cdot \parens*{\Theta\parens*{\frac k \ell}}^{\ell - 2t} \cdot 2^{\ell - 2t} \\
&= \parens*{\Theta\parens*{\frac \ell k}}^{2t} \cdot \parens*{\Theta\parens*{k}}^t = \parens*{\Theta\parens*{\frac {\ell^2} k}}^{t}.
\end{align*}

In other words, to guarantee that a draw from $\mu_\ell$ contains less than $t$ duplicates with probability at least $1 - \delta$, we need to ensure that
\[\parens*{\Theta\parens*{\frac {\ell^2} k}}^{t} \le \delta \implies \ell = O\parens*{\sqrt{k} \delta^{\frac 1 {2t}}}.\]
For $\delta$ scaling inverse-polynomially in $k$ (which is necessary to perform a union bound over the $\text{poly}(k)$ iterations of rejection sampling), this shows we must take $\ell \le k^{\frac{1}{2} - c}$ for some constant $c$ which depends on our budget constant $B$ from the earlier discussion.

%% file: unsized-dpp.tex
\section{Refined guarantees for bounded symmetric DPPs}\label{sec:symdppbound}

For symmetric PSD ensemble matrices $L$ with non-trivial eigenvalue or trace bounds, we give the following refined result improving upon \cref{thm:sym DPP} in various interesting parameter regimes. 

\begin{theorem}\label{thm:sym DPP with bounded eigenvalue}
	Let $L$ be a $n \times n$ symmetric PSD matrix and $\eps \in (0, 1)$. Let  $\mu: 2^{[n]} \to \R_{\geq 0}$ be the DPP defined by $L.$ Let $K = L (I+L)^{-1}\preceq I$ be the kernel of $L$.  There exists an algorithm to approximately sample from within $\epsilon$ total variation distance of $\mu$ in \[\Otilde*{\min\set*{\sqrt{\tr(K)}, \lambda_{\max}(K) \sqrt{n}}}\] 
parallel time using $\poly(n) (\frac 1 \eps)^{o(1)}  $ machines.
\end{theorem}

We use \cref{alg:filter sample}, a ``filtered'' variant of \cref{alg:batched sample}.

\begin{Algorithm}[ht!]
	\KwIn{DPP $\mu: 2^{[n]}\to \R_{\geq 0}$ with kernel $K, \lambda_{\max}(K) \leq \lambda$}
	
	$\alpha \leftarrow (\lambda \sqrt{n})^{-1}$\;
		\If{$\alpha > 1$}{
    	(1): Sample $S \sim \mu$ and return $S.$\;
    	}
    	
	$ S_{-1}, K^{(0)}, L^{(0)} \leftarrow \emptyset, K, L$\;
    	\For{$i=0,1,\dots, R$}{
    	(2): Sample $T_i \sim$  DPP with kernel $\tilde{K}^{(i)}: = \alpha K^{(i)} $\;
    	Update $S_i \leftarrow S_{i-1}\cup T_i $\;
    	Update $L^{(i+1)} \leftarrow ((1-\alpha) L^{(i)})^{T_i} $ (where $(L)^T$ is the ensemble matrix corresponding to the DPP with ensemble matrix $L$, conditioned on including $T$; see \cref{ssec:dppdef})\;
    	Update $K^{(i+1)} \leftarrow I - (I+L^{(i+1)})^{-1}$\;
    	}
    	
    Output $S_R$
    \caption{Filtering} \label{alg:filter sample}
	\end{Algorithm}

We prove \cref{thm:sym DPP with bounded eigenvalue} in this section. Our first step is to show that for $R=  \Theta(\alpha^{-1} \log\frac n \eps)$, the output distribution of \cref{alg:filter sample} is within $\epsilon$ of the target distribution $\mu.$ 

We require the following helper claims. The first shows that randomly independently dropping elements of a sample from a DPP $\mu$ is equivalent to scaling the kernel matrix.

\begin{proposition} \label{prop:filter equivalent}
Let $\mu$ be a DPP with kernel $K.$ Let $\mu'$ be the DPP with kernel $K' := \alpha K.$ Let $\nu$ be the distribution obtained by first sampling $U \sim \mu,$ then outputting $S\subseteq U$ with probability $ \alpha^{\abs{S}} (1-\alpha)^{\abs{S}}$, i.e.\
\[\nu(S) = \sum_{U\supseteq S} \mu(U) \alpha^{\abs{S}} (1-\alpha)^{\abs{U}-\abs{S}}.\]
Then, $\mu'$ and $\nu$ are identical.
\end{proposition}
\begin{proof}
Given a set $A,$ we have $\P_{S\sim \mu'} {A \subseteq S} = \det((\alpha K)_A) = \alpha^{\abs{A}} \det(K_A)= \alpha^{\abs{A}} \sum_{U\supseteq A} \mu(U). $ On the other hand, we have
\begin{align*}
    \P_{S \sim \nu} {A \subseteq S}&= \sum_{S\supseteq A}\nu(S) =   \sum_{U\supseteq S\supseteq A} \mu(U) \alpha^{\abs{S}} (1-\alpha)^{\abs{U}-\abs{S}}\\
    &=\alpha^{\abs{A}}\sum_{U\supseteq S\supseteq A} \mu(U)\alpha^{\abs{S}-\abs{A} } (1-\alpha)^{\abs{U}-\abs{S}}\\
    &= \alpha^{\abs{A}} \sum_{U \supseteq A} \mu(U) \sum_{S' \subseteq U \setminus A} \alpha^{\abs{S'}} (1-\alpha)^{\abs{U\setminus A}-\abs{S'}}\\
    &=  \alpha^{\abs{A}}\sum_{U \supseteq A} \mu(U).\qedhere
\end{align*}
\end{proof}
\begin{proposition} \label{prop:dpp bounded eigenvalue output}
Consider the setup of \cref{alg:filter sample}. Suppose $\alpha\leq 1.$ Let $\P_i{}$ denote the distribution of $S_i.$ 
Fix $\epsilon > 0.$ For $i = \Omega(\alpha^{-1} \log \frac{n}{\epsilon} )$ for a sufficiently large constant, $\dTV{\P_i{}, \mu} \leq \epsilon$.
\end{proposition}

\begin{proof}
Let $\mu^{(i)}$ be the DPP with ensemble matrix $L^{(i)}$ (and kernel matrix $K^{(i)}$), and let $\nu^{(i)}$ be the DPP with kernel matrix $\alpha K^{(i)}.$ We will prove by induction that for all $i$,
\[\P_i{S_i} = \sum_{U \supseteq S_i}   \mu^{(0)} (U) (1-\alpha)^{(i+1)(\abs{U}-\abs{S_i})} (1- (1-\alpha)^{i+1})^{\abs{S_i}}. \]
The base case $i=0$ follows from \cref{prop:filter equivalent}. Now, supposing the induction hypothesis holds for some $i-1,$ we show that it also holds for $i.$ In the following, let $S_0$ be the set sampled in the first iteration of \cref{alg:filter sample}, and let $\P_i{S_i \given S_0}$ denote the probability we observe $S_i$ conditioned on the value of $S_0$. The induction hypothesis then yields the probability we observe $S_i \setminus S_0$ in the next $i - 1$ iterations, with the starting matrix $L^{(1)} \gets ((1 - \alpha)L)^{S_0}$ as follows:
\[\P_i{S_i \given S_0} = \sum_{U \supseteq S_i} \mu^{(1)}(U \setminus S_0) (1 - \alpha)^{i(|U\setminus S_0| - |S_i \setminus S_0|)} (1 - (1 - \alpha)^i)^{|S_i \setminus S_0|}.\]
Hence, we compute 
\begin{align*}
    \P_i{S_i} &= \sum_{S_0\subseteq S_i}\P_i{S_i \given S_0} \P_0{S_0} \\
    &= \sum_{U \supseteq S_i\supseteq S_0}\mu^{(1)} (U\setminus S_0) \P_0{S_0} \\
    &\cdot (1-\alpha)^{i(\abs{U}-\abs{S_i })} (1- (1-\alpha)^{i})^{\abs{S_i \setminus S_0}} \\
    &= \sum_{U \supseteq S_i\supseteq S_0}  \mu^{(0)} (U) (1-\alpha)^{\abs{U}-\abs{S_0}} \alpha^{\abs{S_0}} \\
    &\cdot (1-\alpha)^{i(\abs{U}-\abs{S_i })} (1- (1-\alpha)^{i})^{\abs{S_i \setminus S_0}} \\
    &= \sum_{U \supseteq S_i} \mu^{(0)} (U) (1-\alpha)^{(i+1)(\abs{U} -\abs{S_i})} \\
    &\cdot \sum_{S_0 \subseteq S_i} (1-\alpha)^{\abs{S_i} -\abs{S_0}} (1- (1-\alpha)^{i})^{\abs{S_i \setminus S_0}}  \alpha^{\abs{S_0}}\\
    &= \sum_{U \supseteq S_i} \mu^{(0)} (U) (1-\alpha)^{(i+1)(\abs{U} -\abs{S_i})} \\
    &\cdot \parens*{\alpha + (1-\alpha) (1- (1-\alpha)^{i}) }^{\abs{S_i}}\\
    &= \sum_{U \supseteq S_i} \mu^{(0)} (U) (1-\alpha)^{(i+1)(\abs{U} -\abs{S_i})} (1-(1-\alpha)^{i+1})^{\abs{S_i}}
\end{align*}
where the third equality uses \cref{prop:filter equivalent} and the definition of $L_1 = ((1-\alpha) L)^{S_0}$ to derive
\begin{align*}
    \mu^{(1)} (U\setminus S_0) \P_0{S_0} &= \mu^{(1)} (U\setminus S_0) \sum_{V\supseteq S_0} (1-\alpha)^{\abs{V \setminus S_0}} \alpha^{\abs{S_0}} \mu^{(0)} (V )  \\
    &= \frac{(1-\alpha)^{\abs{U\setminus S_0}} \mu^{(0)} (U ) }{\sum_{V\supseteq S_0} (1-\alpha)^{\abs{V \setminus S_0}} \mu^{(0)} (V )  } \\
    &\cdot \sum_{V\supseteq S_0} (1-\alpha)^{\abs{V \setminus S_0}} \alpha^{\abs{S_0}}  \mu^{(0)} (V ) \\
    &=  (1-\alpha)^{\abs{U\setminus S_0}}  \alpha^{\abs{S_0}}  \mu^{(0)} (U ).
\end{align*}
Thus the induction hypothesis holds for all $i.$ By taking $i = \Omega(\tfrac{\log \frac{n}{\epsilon}}{\alpha} )$, and only considering the summand corresponding to $U = S_i$,
\[\P_i{S_i} \geq \mu(S_i) (1-(1-\alpha)^{i+1})^{\abs{S_i}} \geq  \mu(S_i) \parens*{1 - O\parens*{\frac{\epsilon}{n}}}^n \geq \mu(S_i)(1-\epsilon),\]
and hence,
\begin{align*}\dTV{\P_i{\cdot }, \mu} &= \sum_{S_i: \P_i{S_i}\leq \mu(S_i) } (\mu(S_i)- \P_i{S_i}) \\
&\leq \sum_{S_i: \P_i{S_i}\leq \mu(S_i) } \epsilon \mu(S_i) \leq \epsilon .\qedhere \end{align*}
\end{proof}
Next, we show that each step of the for loop can be implemented in constant parallel time. 
\begin{lemma} \label{lem:DPP boudned eigenvalue one step sample}
Let $\mu: 2^{[n]} \to \R_{\geq 0}$ be a DPP with marginal kernel $K.$
If $ \lambda_{\max}(K)\leq \frac{1}{\sqrt{n}} $, we can sample from a distribution $\epsilon$-away in total variation from $\mu$ in $\Otilde{1}$ time using $O(\poly(n) (\frac 1 \eps)^{o(1)})$ machines.
\end{lemma}
\begin{proof}
Let $s = c \sqrt{n \log\frac{1}{\epsilon'}}$ for $c$ as in \cref{cor:concentration of size}.
Let $\Omega:= \set*{S \subseteq [n] \given \abs{S} \leq s}. $ Let $p_i: = K_{i,i} = \P_{S \sim \mu}{i\in S}.$ Let $\nu$ be the distribution obtained by independently sampling independent $ b_i \sim \text{Ber}(p_i)$ for all $i \in [n]$, and outputting $T =\set*{i \given b_i = 1}. $ By \cref{cor:concentration of size},  $\sum_{S \in \Omega} \mu(S)\geq 1-\epsilon'$.
Moreover, for fixed $T \in \Omega,$ we have
\begin{align*}
    \frac{\mu(T)}{\nu(T)} &= \frac{\det{L_T} }{\det(I+L)} \parens*{\prod_{i\in T} p_i \prod_{i\not \in T} (1-p_i)}^{-1} \\&= \det(L_T) \det(I-K) \parens*{\prod_{i\in T} K_{i,i} \prod_{i\not \in T} (1-K_{i,i})}^{-1},
\end{align*}
where we use $ I+L = (I-K)^{-1}.$ By applying \cref{cor:hadamard ineq} to $I-K$, we have
\[\det(I-K) \leq \prod_{i \in [n]} (1 - K_{i, i}) \le \prod_{i\not \in T} (1-K_{i,i}), \]
so it suffices to show
\[\det(L_T)  \leq \parens*{\frac 1 \eps}^{o(1)} \prod_{i\in T} K_{i,i} \implies \frac{\mu(T)}{\nu(T)} \le \parens*{\frac 1 \eps}^{o(1)}, \]
at which point we can apply \cref{prop:modified rejection sampling}.
Let $K = U D U^{\intercal}$ where $U\in \R^{n \times n}$ is an orthonormal basis of eigenvectors of $K$, and $D = \diag(\{\lambda_i\}_{i\in [n]})$, where $\lambda_1 \geq \cdots\geq  \lambda_n$ are the eigenvalues of $K.$ By \eqref{eq:K to L}, we can write
\[L = U \diag\parens*{\set*{\frac{\lambda_i}{1-\lambda_i}}_{i\in [n]}}U^{\intercal}. \]
Thus, by applying the Cauchy-Binet formula twice,
\begin{align*}
    \det(L_T) &=\det\parens*{ U_{T,[n]} \diag\parens*{\set*{\frac{\lambda_i}{1-\lambda_i}}_{i\in [n]}}U^{\intercal}_{[n], T})} \\
    &= \sum_{S \subseteq [n], \abs{S} = \abs{T}} \det(\Lambda_{T,S}) \parens*{\prod_{i\in S} \frac{\lambda_i}{1-\lambda_i}} \det(\Lambda^{\intercal}_{S,T})\\
    &\leq \exp\parens*{c\sqrt{\log \frac 1 \eps} } \sum_{S \subseteq [n], \abs{S} = \abs{T}} \det(\Lambda_{T,S}) \parens*{\prod_{i\in S} \lambda_i} \det(\Lambda^{\intercal}_{S,T})\\
    &= \exp\parens*{c\sqrt{\log \frac 1 \eps} } \det(K_T)\\\
    &\leq \exp\parens*{c\sqrt{\log \frac 1 \eps} } \prod_{i\in T} K_{i,i}
\end{align*}
where in the first inequality, we used
\[\prod_{i\in S} (1-\lambda_i) \geq \parens*{1-\frac{1}{\sqrt{n}}}^{\abs{S}} \geq \parens*{1-\frac{1}{\sqrt{n}}}^{s} \geq \exp\parens*{ - c\sqrt{\log \frac 1 \eps} }\]
and in the last inequality, we used \cref{cor:hadamard ineq}.
Thus by \cref{prop:modified rejection sampling}, we can sample from $\tilde{\mu}$ that is $\epsilon$-away from $\mu$ in $O(1)$ parallel time using $ O((\frac{1}{\epsilon})^{o(1)} \poly(n))$ machines, by setting $\delta \gets \frac \eps 2$ and adjusting the definition of $\eps$ in this proof by a constant.
\end{proof}
\begin{proposition}\label{prop:DPP bounded eigenvalue for loop}
Consider the setup of \cref{alg:filter sample}. Suppose $\alpha \leq 1$. We have $\lambda_{\max}(K^{(i)}) \leq \lambda$ for all $i$, so that $ \lambda_{\max} (\tilde{K}^{(i)}) \leq \frac{1}{\sqrt{n}}.$
Consequently, each iteration of the for loop can be implemented in $\Otilde{1}$ parallel time using $ O(\poly(n) (\frac 1 \eps)^{o(1)}) $ machines, up to total variation distance $\eps$.
\end{proposition}
\begin{proof} 
We inductively show that $\lambda_{\max}(K^{(i)}) \leq \lambda$ for all $i$. The base case $i = 0$ directly follows from the input assumption. Now let us assume that $\lambda_{\max}(K^{(i)}) \leq \lambda$ for some $i \geq 0$. We will show that $\lambda_{\max}(K^{(i+1)}) \leq \lambda$ follows. 
Let $S$ be the index set of $L^i$ and $\tilde{S} = S \setminus T_i$ where $T_i$ was sampled. Then,
\[L^{(i+1)} = ((1-\alpha)L^{(i)})^{T_i}
= (1-\alpha)L_{\tilde{S}} - (1-\alpha)L_{\tilde{S}, T_i} L_{T_i, T_i}^{-1} L_{T_i, \tilde{S}}.\]
Since $L^{(i)}$ is PSD, $L^{(i)}_{\tilde{S}, T_i} \parens*{L^{(i)}_{T_i, T_i}}^{-1} L^{(i)}_{T_i, \tilde{S}} \succeq 0$. Thus $L^{(i+1)} \preceq (1-\alpha)L^{(i)}_{\tilde{S}} \preceq L^{(i)}_{\tilde{S}}$. 

Let $\Lambda$ be the set of the eigenvalues of $K^{(i)}$. Due to \eqref{eq:K to L}, the eigenvalues of $L^{(i)}$ are given by $\{ \frac \lam {1 - \lam} \mid \lam \in \Lambda\}.$ As $\frac \lam {1 - \lam}$ is strictly increasing in the range $[0,1]$, and the largest eigenvalue of $L^{(i + 1)}$ is dominated by the largest eigenvalue of $L^{(i)}$ (since restrictions to index sets can only decrease quadratic forms), we have the first desired conclusion. The second conclusion follows from \cref{lem:DPP boudned eigenvalue one step sample} as the eigenvalue bound is satisfied.
\end{proof}
Finally, we are ready to prove \cref{thm:sym DPP with bounded eigenvalue}.
\begin{proof}[Proof of \cref{thm:sym DPP with bounded eigenvalue}]
The bound involving $\tr{K}$ follows from a similar argument as in \cref{cor:reduce DPP to k-DPP}. Note that $ \tr{K} = \E_{S\sim \mu} {\abs{S}},$ and that by \cref{cor:concentration of size}, the set $\Omega: = \set*{S\subseteq [n]\given \abs{S} \leq \tr{K} \log \frac{2}{\epsilon} }$ has $\mu(\Omega) \geq  1-\frac \eps 2.$ When drawing $k$ from $\mathcal{H}$ (the distribution on cardinality values), if $k \leq \tr{K} \log \frac{2}{\epsilon} ,$ we use \cref{thm:sym DPP} to approximately sample from within $\frac \eps 2$ of $\mu_k,$ else we output an arbitrary subset. By the triangle inequality, the output's distribution is within $\frac \eps 2 + \frac \eps 2 = \epsilon$ of $\mu.$ The algorithm runs in the stated parallel time depending on $\tr{K}$ using the number of machines as \cref{thm:sym DPP}.

Now we focus on the bound involving $ \lambda_{\max} (K).$
If $\alpha > 1$ then the conclusion follows from \cref{lem:DPP boudned eigenvalue one step sample} applied to step (1).
Else, suppose $\alpha \leq 1.$
We run \cref{alg:filter sample} with $R = O(\lambda \sqrt{n} \log \frac{n}{\epsilon})$ such that \cref{prop:dpp bounded eigenvalue output} guarantees that if we can run the algorithm correctly, the output has total variation $\frac \eps 2$. Let $\epsilon' = \frac{\epsilon}{R}$. Let $\nu^{(i)}$ be the target distribution of $T_i$ in $i^{\text{th}}$ step of the for loop. By \cref{prop:DPP bounded eigenvalue for loop}, we can modify step (2) to sample from $\hat{\nu }^{(i)}$ that is $\epsilon'$-away from $\nu^{(i)}$ in TV-distance in $\Otilde{1}$ time using $O(\poly(n)(\frac 1 \eps)^{o(1)})$ machines. Hence, by the triangle inequality, the output of the algorithm is $\frac \eps 2$ away from the output if we were given exact sample access to each $\nu^{(i)}$. Combining with the approximation error of \cref{prop:dpp bounded eigenvalue output} yields the conclusion.
\end{proof}

%% file: proofs.tex
\section{Deferred proofs}\label{sec:proofs}

\begin{proof}[Proof of \cref{lem:klrenyi}]
We first prove the inequality for the case $S = [n]$ and $C=1.$
Clearly,
\[f'(r) = n^{\lam-1}\lam \parens*{r^{\lam-1} - (\lam-1) (1 + \log r + \log n)} \]
and 
\[f''(r) =  n^{\lam-1}\lam (\lam-1) \parens*{r^{\lam-2} - \frac{1}{r}} \leq 0. \]
Thus $f$ is concave and
\[\frac{1}{n}\sum_{i=1}^n f(q_i)\leq f \parens*{\frac{1}{n}\sum_{i=1}^n q_i} = f\parens*{\frac 1 n} = \frac{1}{n} \]
which is equivalent to
\[\sum_{i=1}^n q_i^{\lam} n^{\lam-1} \leq 1 + n^{\lam-1} \lam(\lam-1) \sum_{i=1}^n q_i \log (n q_i). \]
The case $C  > 1$ then follows from
\begin{align*}
   \D_{\lam}{q \river p} 
  &= \sum_{i=1}^n q_i^{\lam}p_i^{1-\lam}\\ &\leq C^{\lam-1} \sum_{i=1}^n q_i^{\lam}n ^{\lam-1} \\
  &\leq C^{\lam-1} \parens*{1+ n^{\lam-1} \lam(\lam-1) \sum_{i=1}^n q_i \log (n q_i)}\\
  &\leq C^{\lam-1} \parens*{1+ n^{\lam-1} \lam(\lam-1)\parens*{ \sum_{i=1}^n q_i \log \frac{ q_i}{p_i} + \log C}}.
\end{align*}
The case $S \neq [n]$ follows by noticing
\[\sum_{i\in S} q_i \parens*{\frac{q_i}{p_i}}^{\lam-1} \leq \sum_{i=1}^n q_i^{\lam}p_i^{1-\lam} .\qedhere \]
\end{proof}

\begin{proof}[Proof of \cref{prop:fast compute marginal}]
First, note that DPPs and $k$-DPPs correspond to Partition-DPPs with $0$ and $1$ partition constraints respectively. We thus show the claim for $\mu$ a Partition-DPP with $O(1) $ constraints.
Computing the marginals is equivalent to computing the partition functions of $\mu$ and $\mu$ conditioned on subsets.
As shown in \cite[Theorem 1.1]{Celis2017OnTC}, computing the partition function is equivalent to computing the coefficients of a polynomial in $O(1)$ variables, obtained from the generating polynomial by plugging the same variable for all elements in a partition. This polynomial can be evaluated efficiently given access to $g_\mu$. Evaluating $g_{\mu}$ at $(z_1,\dots, z_n)$ is equivalent to computing $\det(L + \diag(z_i)_{i=1}^n),$ which can be done in $\Otilde{1}$-parallel time \cite{Berkowitz84}. From evaluations we can obtain the coefficients by polynomial interpolation; this can be solved, e.g., by solving linear equations, doable in \Class{NC} \cite{csanky1975fast}.
\end{proof}

\begin{proof}[Proof of \cref{cor:concentration of size}]

Let $f(S) = \card{S}$ correspond to the $1$-Lipschitz (with respect to the Hamming metric) function of the size of a sample $S$ from $\mu$, and let $\overline{f}$ indicate the value $\E_{S \sim \mu}{f}$. 

By applying known concentration inequalities for strongly Rayleigh distributions \cite[Theorem 3.2]{pemantle2013concentration}, it follows that \[\P_{S\sim \mu} {f - \overline{f} > a} \leq 3 \text{ exp } \bigg{(} \frac{-a^2}{16(a + 2\overline{f})} \bigg{)}.\]

For the first statement, let $a = 10 \sqrt{n \log\frac{1}{\eps}},$ and note $\exp(-\frac{a^2 }{16(a + 2 \bar{f})}) \leq \frac \eps 3$ and $a + \bar{f} \leq 11 \sqrt{n \log\frac{1}{\eps}}$. For the second, let $a= \bar{f} \log\frac{1}{\eps} $  and note $ \exp(-\frac{a^2 }{16(a + 2 \bar{f})}) \leq \frac \eps 3$ and $a + \bar{f} \leq 2 \bar{f}\log\frac{1}{\eps}.$
\end{proof}

\begin{proof}[Proof of \cref{prop:near-isotropic}]
First, we check the cardinality of the new ground set $U$:
\begin{align*}
n \beta^{-1} &=\sum_{i=1}^n \frac{n}{k\beta }p_i \leq\card{U} \\
&= \sum_{i=1}^n t_i\leq \sum_{i=1}^n\parens*{1+\frac{n}{k\beta }p_i}=n+\frac{n}{k\beta }\sum_{i=1}^n p_i=n(1+ \beta^{-1}).
\end{align*}

Next, we check that for any $i^{(j)}$, the marginal probabilities $\P_{S \sim \mu^{\text{iso}}}{i^{(j)} \in S}$ are at most $\frac{Ck}{|U|}$. In the following calculation, we interpret sampling from $\mu^{\text{iso}}$ as first sampling from $\mu$ and then choosing a copy $j \in [t_i]$ for each element. This yields
\begin{align*}
\P_{S \sim \mu^{\text{iso}}}{i^{(j)} \in S} &= \sum_{S \ni i} \P{\text{we chose copy } j \given \text{we sampled } S \text{ from } \mu} \\
&\cdot \P{\text{we sampled } S \text{ from } \mu} \\
&= \sum_{S \ni i} \frac{1}{t_i} \cdot \mu(S) = \frac{1}{t_i} \sum_{S \ni i} \mu(S) = \frac{1}{t_i} \cdot \P_{S \sim \mu}{i \in S}.
\end{align*}

Since $ 1+\frac{n}{\beta k}p_i \ge t_i\geq \frac{n}{\beta k}p_i$, we obtain
\begin{align*} \frac{k}{ k p_i^{-1} + n \beta^{-1} } &= \frac{p_i}{1 + \frac n {\beta k} p_i}\le \P_{S\sim \mu^{\text{iso}}}{i^{(j)}\in S}\\
&\leq \frac{\beta k}{ n} = \frac{k (\beta+1)}{n(1+\beta^{-1})}\leq \frac{k(\beta +1) }{\abs{U} } \leq \frac{C k}{\abs{U}}.\end{align*} 
The latter inequality shows the marginal upper bound. Next, to show the marginal lower bound, suppose $ \P_{\mu}{i\in S}= p_i \geq \frac{\sqrt{\beta} k}{n }$. Then for all $j \in [t_i]$,
\[ \P_{S\sim \mu^{\text{iso}}}{i^{(j)}\in S} \geq \frac{k}{ k p_i^{-1} + n \beta^{-1} } \geq\frac{k}{n \beta^{-1} (1 + \sqrt \beta) } \geq \frac{k}{C\abs{U}}.\]
Finally, letting $\bar{R} \coloneq \set*{i \given p_i \geq \frac{\sqrt \beta k}{n}} \subseteq [n]$,
\begin{align*}\sum_{S\in \binom{R}{\l}} \mu^{\text{iso}}_\l(S) &=\sum_{\bar{S}\in \binom{\bar{R}}{\l}} \mu_\l(\bar{S}) \\&= 1 - \sum_{\bar{S} \subseteq \binom{[n]}{\ell}: \bar{S} \not\subseteq \binom{\bar{R}}{\ell}} \mu_\l(\bar{S})  \\
&\geq 1- \sum_{i\not\in \bar{R}} \sum_{\bar{S}\subseteq \binom{[n]}{\ell}: i \in \bar{S} }   \mu_\l(\bar{S})   \\
&= 1- \sum_{i\not\in \bar{R}} \frac{\l  p_i}{k} \geq 1- \sqrt{\beta} \l.\qedhere \end{align*}
\end{proof}

\begin{proof}[Proof of \cref{lem:klcondmarg}]
Throughout this proof, fix the set $S \in \binom{[n]}{t}$, and let $A_S$ denote the left-hand side of \eqref{eq:klcondmarg}. Let $q_i \coloneq \P_{T \sim \mu}{i \in T \given S \subseteq T}$, and note that $q_i = 1$ for all $i \in S$. Moreover, we have
\begin{equation}\label{eq:abound}
\begin{aligned}
A_S &= \sum_{i \not\in S} \frac{q_i}{k - t} \log\parens*{\frac{q_i}{p_i} \cdot \frac{k}{k - t}} \\
&= \frac{k}{k - t}\sum_{i \not\in S} \frac{q_i}{k} \log\parens*{\frac{q_i}{p_i}} + \log \frac k {k - t} \\
&\le \frac{k}{k - t}\sum_{i \not\in S} \frac{q_i}{k} \log\parens*{\frac{q_i}{p_i}} + \frac{2t}{k}.
\end{aligned}
\end{equation}
The first equation used that by definition, $\mu_{t + 1 \mid S} = \frac{q_{S^c}}{k - t}$ where $q_{S^c}$ restricts $q$ to $S^c \coloneq [n] \setminus S$; the only inequality used $\log(1 + c) \le c$ for all $c \ge 0$ and $t \le \frac{1}{2} k$. We note that for
\[B_S \coloneq \sum_{i \not\in S} \frac{q_i} k \log\parens*{ \frac{q_i}{p_i}} + \sum_{i \in S} \frac 1 k \log \frac 1 {p_i},\]
we have by $t \le \frac{1}{2} k$ that
\begin{equation}\label{eq:ABbound}A_S \le 2\sum_{i \not\in S} \frac{q_i} k \log \parens*{\frac{q_i}{p_i}} + \frac{2t}{k} \le 2B_S + \frac {2t} k,\end{equation}
since $\log \frac 1 {p_i} \ge 0$ for all $i \in [n]$. We next give an interpretation of the quantity $B_S$. Let $\mu_S$ be the distribution of $T \sim \mu$ conditioned on $S \subset T$, so that
\[\mu_S D_{k \to 1} = \begin{cases}
\frac 1 k & i \in S \\
\frac {q_i} k & i\not\in S
\end{cases}.\]
Notice that $B_S$ is defined to be $\DKL{\mu_S D_{k \to 1} \river \mu D_{k \to 1}}$ (since $\mu D_{k \to 1} = \frac 1 k p$), which we can control by entropic independence of $\mu$. In particular, 
\begin{align*}B_S = \DKL{\mu_S D_{k \to 1} \river \mu D_{k \to 1}} &\le \frac{1}{\alpha k} \DKL{\mu_S \river \mu} \\
&= \frac{1}{\alpha k} \sum_{\substack{T \in \binom{[n]}{k} \\ S \subset T}} \mu_S(T) \log \frac{\mu_S(T)}{\mu(T)} \\
&= \frac{1}{\alpha k} \sum_{\substack{T \in \binom{[n]}{k} \\ S \subset T}} \mu_S(T) \\
&\cdot \log\parens*{\frac{\mu(T)}{\P_{T \sim \mu}{S \subset T}} \cdot \frac{1}{\mu(T)}} \\
&= \frac 1 {\alpha k} \log\parens*{\frac 1 {\P_{T \sim \mu}{S \subset T}}}.\end{align*}
Combining the above display with \eqref{eq:ABbound} completes the proof.
\end{proof}

%% file: acmacks.tex
\begin{acks}
	\input{acks}	
\end{acks}